\documentclass[12pt,reqno]{amsart}

\usepackage{graphicx}
\usepackage{amsmath}

\numberwithin{equation}{section}

\newcommand{\R}{{\mathbb R}}

\newcommand{\Z}{{\mathbb Z}}

\newcommand{\al}{\alpha}
\newcommand{\be}{\beta}

\newcommand{\ga}{\gamma}
\newcommand{\Ga}{\Gamma}
\newcommand{\La}{\Lambda}
\newcommand{\la}{\lambda}
\newcommand{\ep}{\varepsilon}
\newcommand{\de}{\delta}

\newcommand{\sg}{\sigma}
\newcommand{\Sg}{\Sigma}

\newcommand{\Tr}{\textrm{Tr}\,}
\newcommand{\Pf}{{\operatorname{Pf}\,}}
\newcommand{\sgn}{{\operatorname{sgn}\,}}

\newtheorem{theo}{{\sc \bf Theorem}}[section]

\newtheorem{lem}[theo]{{\sc \bf Lemma}}
\newtheorem{prop}[theo]{{\sc \bf Proposition}}

\begin{document}

\title[The Pfaffian Sign Theorem for the Dimer Model on a Triangular Lattice]
{The Pfaffian Sign Theorem for the Dimer Model on a Triangular Lattice}

\author{Pavel Bleher, Brad Elwood, and Dra\v zen Petrovi\'c}

\date{\today}

\thanks{The first author is supported in part
by the National Science Foundation (NSF) Grants DMS-1265172 and DMS-1565602.}

\begin{abstract}
We prove the Pfaffian Sign Theorem for the dimer model on a triangular lattice embedded in the torus. More specifically, we prove that the Pfaffian of the Kasteleyn periodic-periodic matrix is negative, while the Pfaffians of the Kasteleyn periodic-antiperiodic, antiperiodic-periodic, and antiperiodic-antiperiodic matrices are all positive. The proof is based on the Kasteleyn identities and on small weight expansions. As an application, we obtain an asymptotic behavior of the dimer model partition function with an exponentially small error term.
\end{abstract}

\maketitle

\section{Introduction}

\subsection{Dimer model on a triangular lattice}

We consider the dimer model on a triangular lattice $\Gamma_{m,n}=(V_{m,n},E_{m,n})$ on the torus $\Z_m\times \Z_n=\Z^2/(m\Z\times n\Z)$ (periodic boundary conditions), where
$V_{m,n}$ and $E_{m,n}$ are the sets of vertices and edges of $\Gamma_{m,n},$ respectively. It is convenient to consider $\Gamma_{m,n}$ as a square lattice with diagonals. 
A \textit{dimer} on $\Gamma_{m,n}$ is a set of two neighboring vertices $\langle x,y\rangle$ connected by an edge. 
A \textit{dimer configuration} $\sg$ on $\Ga_{m,n}$ is a set of dimers $\sg=\{\langle x_i,y_i\rangle,\;i=1,\ldots,\frac{mn}{2}\}$ which cover $V_{m,n}$ without overlapping. 
An example of a dimer configuration is shown in Fig.\ \ref{F8}. An obvious necessary condition for a configuration to exist is that at least one of $m,n$ is even, and so we assume that $m$ is even, $m=2m_0$.

\begin{figure}[h!]
\includegraphics[scale=.5]{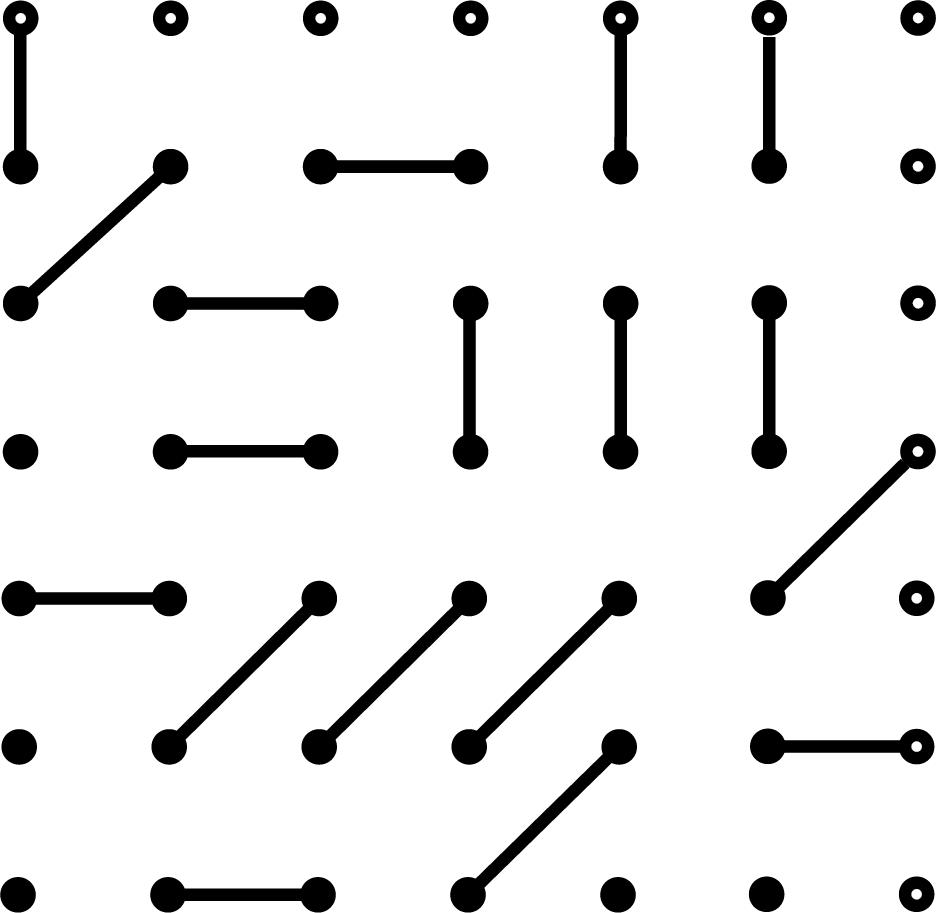}
\caption{Example of a dimer configuration on a triangular $6\times 6$ lattice on the torus.}
\label{F8}
\end{figure}

To define a weight of a dimer configuration, we split the full set of dimers in a configuration $\sg$ into three classes: horizontal, vertical, and diagonal with respective weights $z_h, z_v, z_d>0.$ 
If we denote the total number of horizontal, vertical and diagonal dimers in $\sg$ by $N_h(\sg)$, $N_v(\sg),$ and $N_d(\sg),$ respectively, then the \textit{dimer configuration weight} is
\begin{equation}\label{int1}
w(\sg)=\prod_{i=1}^{\frac{mn}{2}}w(x_i,y_i)=z_h^{N_h(\sg)} z_v^{N_v(\sg)}z_d^{N_d(\sg)},
\end{equation}
where $w(x_i,y_i)$ denotes the weight of the dimer $\langle x_i,y_i\rangle\in\sg$.
We denote by $\Sigma_{m,n}$ the set of all dimer configurations on $\Ga_{m,n}$. 
The \textit{partition function} of the dimer model is given by
\begin{equation}\label{int2}
Z=\sum_{\sg\in\Sigma_{m,n}}w(\sg).
\end{equation}
Notice that if all the weights are set equal to one, then $Z$ simply counts the number of dimer configurations, or perfect matchings, on $\Gamma_{m,n}$. 

\subsection{Main result} As shown by Kasteleyn \cite{Kas1}--\cite{Kas3}, the partition function $Z$ of the dimer model on the torus can be expressed in terms of the
 four Kasteleyn Pfaffians as
\begin{equation}\label{a1}
Z=\frac{1}{2}\left(-\Pf A_{1} +\Pf A_{2} +\Pf A_{3} + \Pf A_{4}\right),
\end{equation}
with periodic-periodic, periodic-antiperiodic,  antiperiodic-periodic, and antiperiodic-antiperiodic boundary conditions
in the $x$-axis and $y$-axis, respectively. For an extension of formula \eqref{a1} to graphs on Riemannian surfaces of higher genera see  the works of Galluccio and Loebl \cite{GalluLoe}, Tesler \cite{Tes}, Cimasoni and Reshetikhin \cite{CimResh}, and references therein. Formula \eqref{a1} is very powerful in the asymptotic analysis of the partition function as $m,n\to\infty$, because the absolute value of the Pfaffian of a square antisymmetric matrix $A$ is determined by its determinant through the classical identity
\begin{equation}\label{a1a}
(\Pf A)^2=\det A.
\end{equation}
The asymptotic behavior of $\det A_i$ as $m,n\to\infty$ can be analyzed by a diagonalization of the matrices $A_i$ (see, e.g., \cite{Kas1}, \cite{Fend}), and
an obvious problem arises to determine the {\it sign} of the Pfaffians $\Pf A_i$
in formula \eqref{a1}. 

In \cite{Kas1} Kasteleyn considered the dimer model on the square lattice, which corresponds to the weight $z_d=0$. He showed that in this case $\Pf A_1=0$ and he assumed that $\Pf A_i\ge 0$ for $i=2,3,4$. Kenyon, Sun and Wilson \cite{Kenyon} established the sign of the Pfaffians $\Pf A_i$ for any critical dimer model on a lattice on the torus, including the square lattice. The dimer model on the triangular lattice is not critical and the results of \cite{Kenyon} are not applicable in this case. Different conjectures about the Pfaffian signs for the dimer model on a triangular lattice are stated, without proof, in the works of McCoy \cite{McCoy}, Fendley, Moessner, and Sondhi \cite{Fend}, and Izmailian and Kenna \cite{IzmaKenna}. 

Our main result in this paper is the following theorem:

\begin{theo}[The
Pfaffian Sign Theorem]\label{PST} Let $z_h,z_v,z_d>0$. Then
\begin{equation}\label{PSTF}
\Pf A_1<0,\quad \Pf A_2>0,\quad \Pf A_3>0,\quad \Pf A_4>0.
\end{equation}
\end{theo}

The proof of this theorem is given below and it is based on the following two important ingredients:

\begin{enumerate}
\item The Kasteleyn formulae for the Pfaffians $\Pf A_i$ in terms of algebraic sums of the partition functions of the dimer model restricted to different $\Z_2$ homology classes.
\item An asymptotic analysis of $\Pf A_i$ as one of the weights tends to zero. It is worth noting that due to various cancellations in the Pfaffians the leading terms in the small weight asymptotics of the Pfaffians $\Pf A_i$ depend on arithmetic properties of the lattice dimensions $m$ and $n$, and it requires a geometric description of configurations giving the leading contribution to the Pfaffians $\Pf A_i$ for different $m,n$.
\end{enumerate}

As an application of Theorem \ref{PST}, we obtain an asymptotic behavior of the partition function $Z$ as $m,n\to \infty$ with an exponentially small error term.
See formula \eqref{AsympF} in Section \ref{Poisson} below. In the works
\cite{Kenyon} and \cite{McCoy} it is stated without a proof that in the noncritical case all the
Pfaffians $A_i$ in formula \eqref{a1} are {\it positive}. Our Theorem \ref{PST}
disproves this statement, and this gives a constant 2 instead of 1   in the asymptotic formula for $Z$.

The set-up for the remainder of the paper is the following. In Section
\ref{Kasteleyn} we review Pfaffians of the  Kasteleyn matrices and their properties. In Section \ref{preliminary} we prove various preliminary
results about the Pfaffians $\Pf A_i$. In Section \ref{A34}
we prove that $\Pf A_3>0$, $\Pf A_4>0$. In Section \ref{Id} we prove the identities $\Pf A_1=-\Pf A_2$ and $\Pf A_3=\Pf A_4$ for odd $n$.  In Section \ref{A2} we prove that
$\Pf A_2>0$, and in Section \ref{A1} that
$\Pf A_1<0$, which is the most difficult part of our study. 
In Section \ref{Poisson} we obtain the asymptotics of the partition function  as $m,n\to\infty$. In Appendix \ref{appA} we prove a sign formula 
for the superposition of two dimer configurations. In Appendix \ref{appB} we prove Kasteleyn identities for the triangular lattice on the torus, and in Appendix \ref{numerics} we present numerical data for the Pfaffians $\Pf A_i$ for different dimensions $m$ and $n$.

\section{Dimer model and Kasteleyn matrices}\label{Kasteleyn}

We consider different orientations on the set of the edges $E_{m,n}$: 
$O_1$ (periodic-periodic),  $O_2$ (periodic-antiperiodic),  
$O_3$ (antiperiodic-periodic), and $O_4$
 (antiperiodic-antiperiodic), depicted in Fig.\ \ref{KOs} for $m=4$, $n=3$. 
\begin{figure}[h!]
\includegraphics[scale=.3]{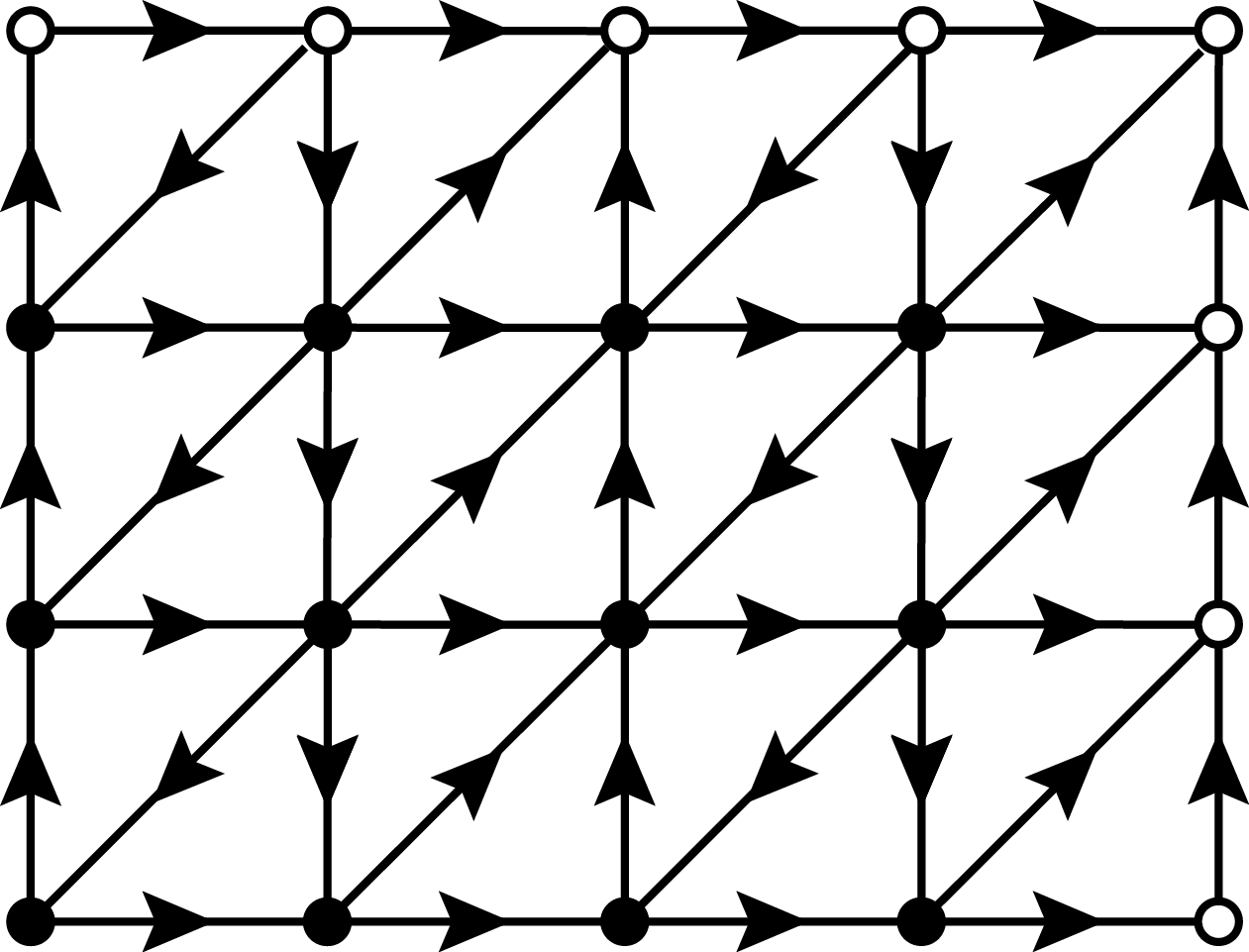}\hspace{.35in}\includegraphics[scale=.3]{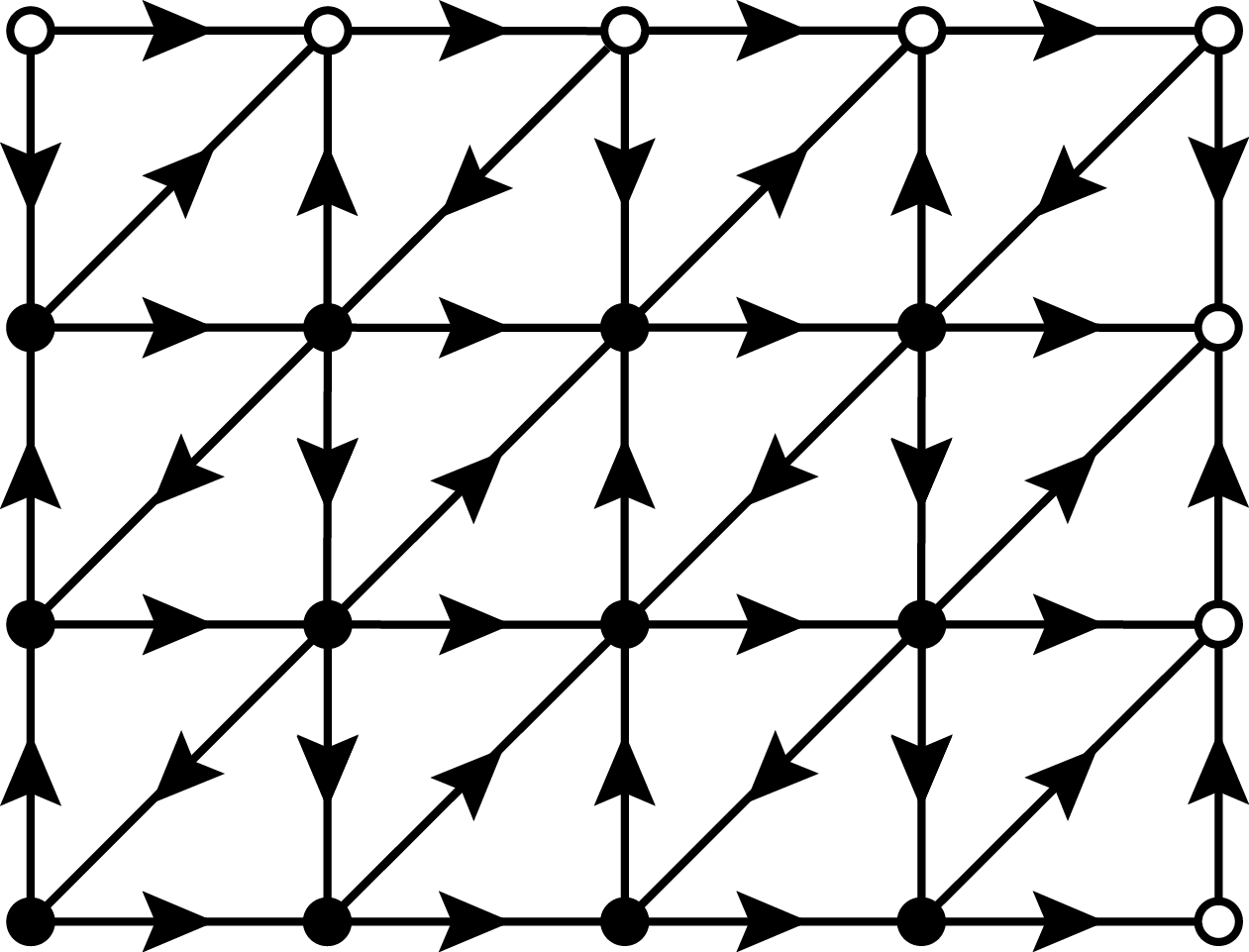}\hspace{.35in}\includegraphics[scale=.3]{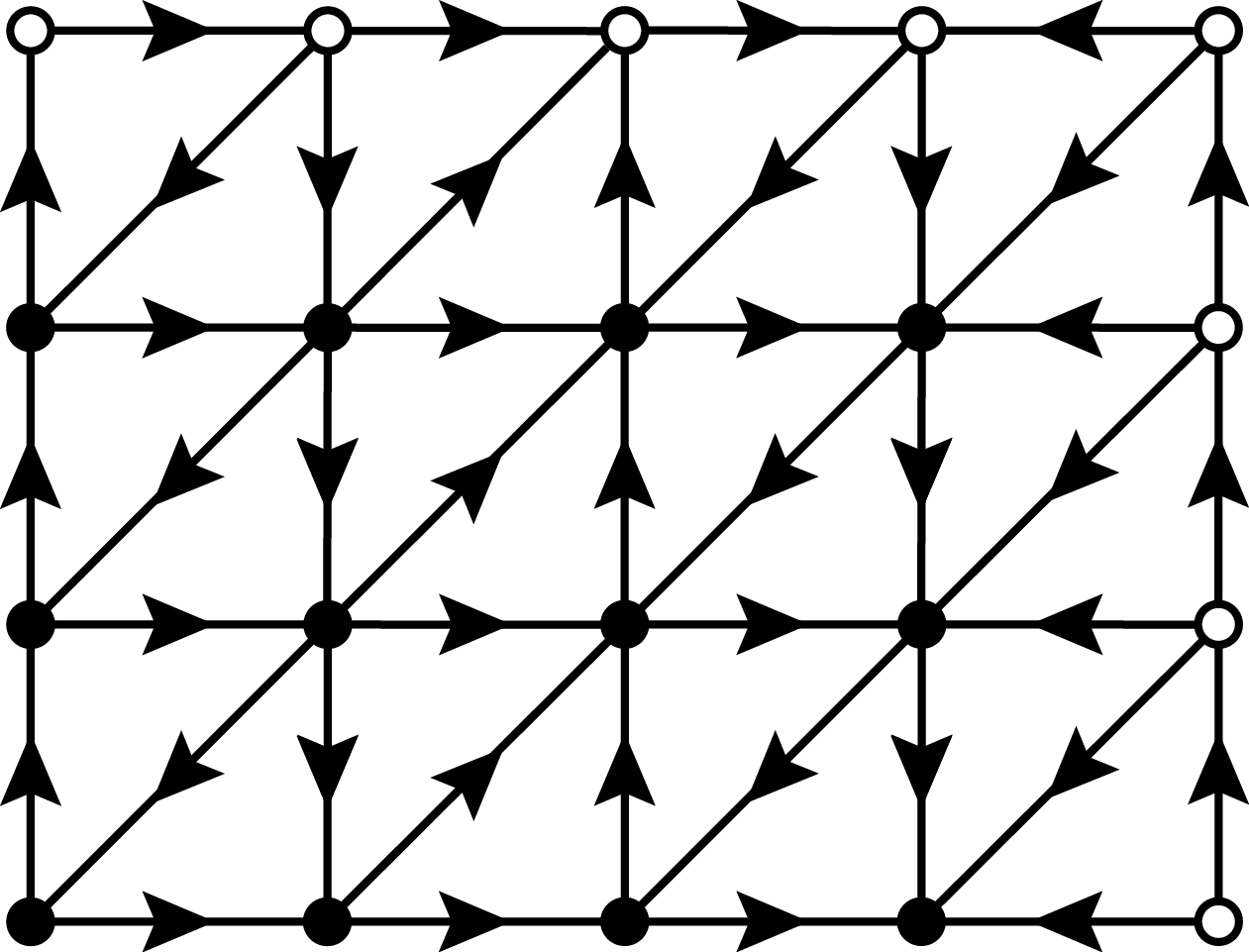}\hspace{.35in}\includegraphics[scale=.3]{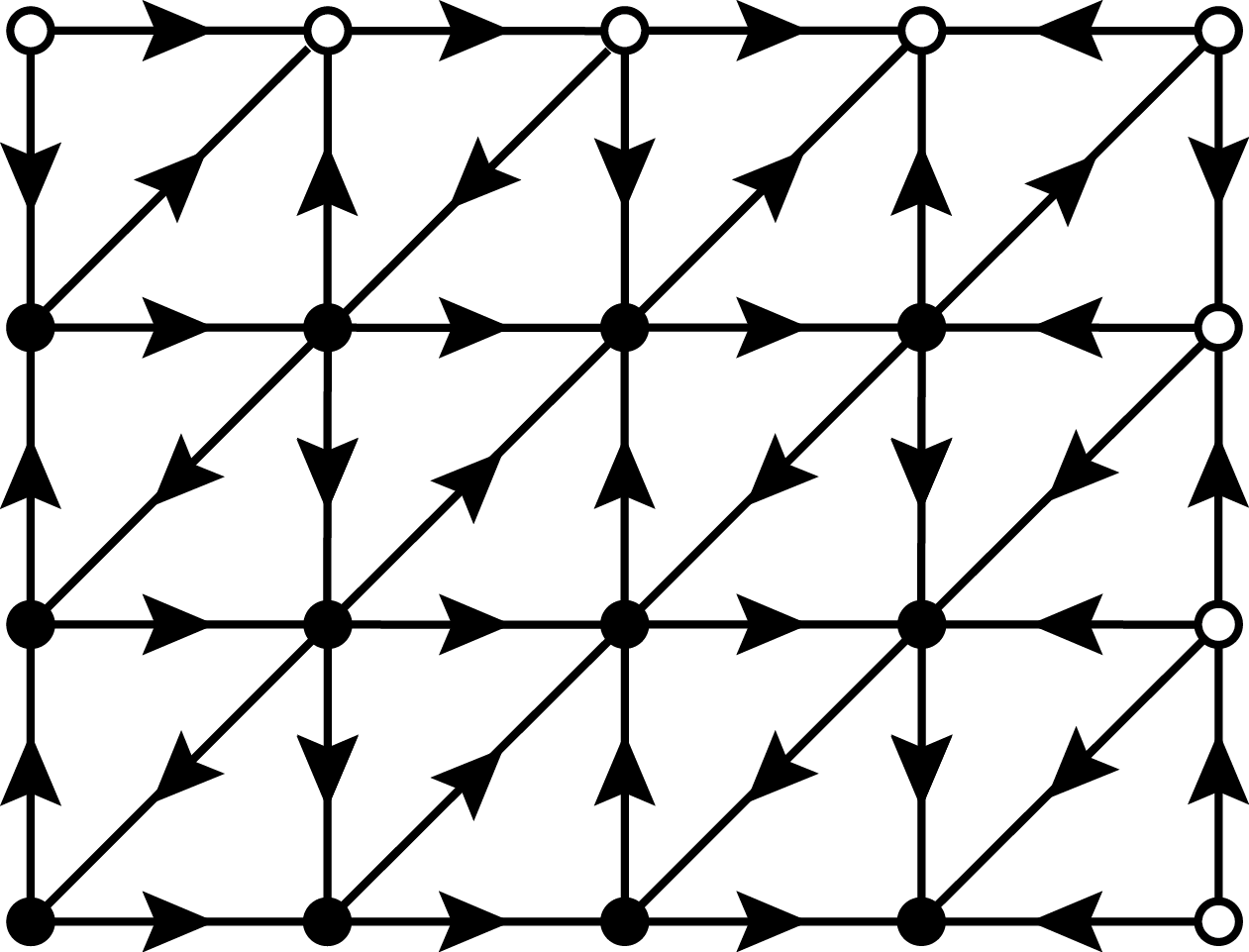}
\caption{The Kasteleyn orientations for $m=4$, $n=3$: $O_1$ (periodic-periodic),  $O_2$ (periodic-antiperiodic),  $O_3$ (antiperiodic-periodic), and
 $O_4$ (antiperiodic-antiperiodic).}
\label{KOs}
\end{figure}
All these orientations are  {\it Kasteleyn orientations}, so that for any face
the number of arrows on the boundary oriented {\it clockwise} is {\it odd}.

With every orientation $O_i$ we associate a sign function $\ep_i(x,y)$,
$x,y\in V_{m,n}$, defined as follows: if $x$ and $y$ are connected by an edge then
\begin{equation}\label{Aj3}
\ep_i(x,y)=\left\{
\begin{alignedat}{2}
&1,\quad &&\textrm{if the arrow in the Kasteleyn orientation $O_i$ points from $x$ to $y$,}\\
& -1, \quad &&\textrm{if the arrow in the Kasteleyn orientation $O_i$ points from $y$ to $x$,}
\end{alignedat}\right.
\end{equation}
and 
\begin{equation}\label{Aj3a}
\ep_i(x,y)=0,\quad \textrm{if $x$ and $y$ are not connected by an edge.}
\end{equation}
More specifically, the sign functions are given by the following formulae. 

Let $\mathsf e_1=(1,0)$, $\mathsf e_2=(0,1)$, and $x=(j,k)\in \Z_m\times \Z_n$. Then  the sign function $\ep_1$ takes the values 
\begin{equation}\label{Aj3b}
\begin{aligned}
&\ep_1(x,x+\mathsf e_1)=1,\\
&\ep_1(x,x+\mathsf e_2)=(-1)^j,\\
&\ep_1(x,x+\mathsf e_1+\mathsf e_2)=(-1)^{j+1}.
\end{aligned}
\end{equation}
The sign function $\ep_2$ is obtained from $\ep_1$ by the reversal of the
vertical and diagonal arrows in the upper row, so that
\begin{equation}\label{Aj3c}
\begin{aligned}
&\ep_2(x,x+\mathsf e_1)=1,\\
&\ep_2(x,x+\mathsf e_2)=(-1)^{j+\de_{k,n-1}},\\
&\ep_2(x,x+\mathsf e_1+\mathsf e_2)=(-1)^{j+1+\de_{k,n-1}}.
\end{aligned}
\end{equation}
Similarly,
the sign function $\ep_3$ is obtained from $\ep_1$ by the reversal of the
horizontal and diagonal arrows in the last column, so that
\begin{equation}\label{Aj3d}
\begin{aligned}
&\ep_3(x,x+\mathsf e_1)=(-1)^{\de_{j,m-1}},\\
&\ep_3(x,x+\mathsf e_2)=(-1)^{j},\\
&\ep_3(x,x+\mathsf e_1+\mathsf e_2)=(-1)^{j+1+\de_{j,m-1}}.
\end{aligned}
\end{equation}
Finally, the sign function $\ep_4$ is obtained from $\ep_1$ by the reversal of  both the vertical and diagonal arrows in the upper row and the
horizontal and diagonal arrows in the last column, so that
\begin{equation}\label{Aj3e}
\begin{aligned}
&\ep_4(x,x+\mathsf e_1)=(-1)^{\de_{j,m-1}},\\
&\ep_4(x,x+\mathsf e_2)=(-1)^{j+\de_{k,n-1}},\\
&\ep_4(x,x+\mathsf e_1+\mathsf e_2)=(-1)^{j+1+\de_{j,m-1}+\de_{k,n-1}}.
\end{aligned}
\end{equation}
In addition, \eqref{Aj3a} holds and $\ep_i(y,x)=-\ep_i(x,y)$.

With every orientation $O_i$ we associate a Kasteleyn matrix $A_i$.
To define the Kasteleyn matrices, consider any enumeration of the vertices,
$V_{m,n}=\{x_1,\ldots,x_{mn}\}$. Then the \textit{Kasteleyn matrices} $A_i$ 
are defined as
\begin{equation}\label{Aj1}
A_i=\big(a_i(x_j,x_k)\big)_{1\le j,k\le mn}\,,\quad i=1,\,2,\,3,\,4,
\end{equation}
with
\begin{equation}\label{Aj2}
a_i(x,y)=\left\{
\begin{alignedat}{2}
&\ep_i(x,y) w(x,y),\quad &&\textrm{if $x$ and $y$ are connected by an edge,}\\
& 0, \quad &&\textrm{otherwise,}
\end{alignedat}\right.
\end{equation}
where $w(x,y)=z_h,z_v,z_d$ is the weight of the dimer $\langle x,y\rangle$
and $\ep_i$ is the sign function. Consider now the Pfaffians $\Pf A_i$.

The \textit{Pfaffian}, $\Pf A_i$, of the $mn\times mn$ antisymmetric matrix $A_i,$ $i=1,\, 2,\, 3,\, 4,$ is a number given by
\begin{equation}\label{idd2}
\begin{aligned}
\Pf A_i&=\sum_{\pi}(-1)^{\pi}a_i(x_{p_1},x_{p_2})a_i(x_{p_3},x_{p_4})\cdots a_i(x_{p_{mn-1}},x_{p_{mn}}),
\end{aligned}
\end{equation} 
where the sum is taken over all permutations, 
\begin{equation}
\pi = \left(\begin{matrix}
  1 & 2 & 3 & \cdots & mn-1 & mn \\
  p_1 & p_2 & p_3 & \cdots &  p_{mn-1}  & p_{mn}
\end{matrix}\right),
\end{equation}
which satisfy the following restrictions:
\begin{enumerate}
\item[(1)] $p_{2\ell-1} < p_{2\ell}, \qquad 1\leq \ell \leq \frac{mn}{2}$,
\item[(2)] $p_{2\ell-1} < p_{2\ell+1}, \quad 1 \leq \ell \leq \frac{mn}{2}-1$.
\end{enumerate}
Such permutations are in a one-to-one correspondence with the dimer configurations,
\begin{equation}\label{bijection1}
\begin{aligned}
\pi = \left(\begin{matrix}
  1 & 2 & 3 & \cdots & mn-1 & mn \\
  p_1 & p_2 & p_3 & \cdots &  p_{mn-1}  & p_{mn}
\end{matrix}\right)\leftrightarrow \sg=\left\{\langle p_{2i-1},p_{2i}\rangle,\; i=1,\ldots,\frac{mn}{2}\, \right\}\,,
\end{aligned}
\end{equation} 
and \eqref{idd2} can be rewritten as
\begin{equation}\label{PfOurDef}
\Pf A_i=\sum_{\sg\in\Sigma_{m,n}}(-1)^{\pi(\sg)}w(\sg)\prod_{\langle x,y\rangle \in\sg}\varepsilon_i(x,y), \quad i=1,2,3,4,
\end{equation}
where $\Sigma_{m,n}$ is the set of all dimer configurations on $\Gamma_{m,n}$ and $(-1)^{\pi(\sg)}$ is the sign of the permutation $\pi(\sigma)$.

An important property of the Kasteleyn Pfaffians $\Pf A_i$ is that they are covariant with respect to the enumeration of the vertices. Namely, 
if $\rho(x)=\{\rho(x_1),\ldots,\rho(x_{mn})\}$
is an enumeration of the vertices of $V_{m,n}$ obtained from the one $x=\{x_1,\ldots,x_{mn}\}$ by a permutation $\rho$, then
\begin{equation}\label{cov}
\Pf A_i(\rho(x))=(-1)^{\rho}\Pf A_i(x), \quad i=1,2,3,4,
\end{equation}
where $(-1)^{\rho}$ is the sign of the permutation $\rho$. See Appendix \ref{appA} below.

 The \textit{sign of a configuration $\sg$}, $\sgn(\sg)=\sgn(\sg;O_i)$, is the following expression:
\begin{equation}\label{sign}
\sgn(\sg)=(-1)^{\pi(\sg)}\prod_{\langle x,y\rangle \in\sg}\varepsilon_i(x,y),
\end{equation}
where $\varepsilon_i(x,y)$ is given by \eqref{Aj3}. Having \eqref{sign}, the Pfaffian formula for a Kasteleyn matrix $A_i$ can be rewritten as 
\begin{equation}\label{Pfsign}
\Pf A_i=\sum_{\sg\in\Sigma_{m,n}}\sgn(\sg)w(\sg),\quad  \sgn(\sg)=\sgn(\sg;O_i).
\end{equation}
Given two configurations $\sg$ and $\sg'$, we consider the double configuration $\sg\cup\sg'$, and we call it the \textit{superposition of $\sg$ and $\sg'$}. In $\sg\cup\sg'$, we define a \textit{contour} to be a cycle consisting of alternating edges from $\sg$ and $\sg'$. Each contour consists of an even number of edges. The superposition $\sg\cup\sg'$ is partitioned into disjoint contours $\{\ga_k\ :\ k=1,2,\ldots,r\}$. We will call a contour consisting of only two edges a \textit{trivial contour}. 

Let us introduce a standard configuration $\sg_{\rm st}$ as follows. Consider the lexicographic ordering of the
vertices $(i,j)\in \Z_m\times \Z_n$. Namely,
\begin{equation}\label{lex1}
(i,j)=x_k,\quad  k=jm+i+1, \quad 1\le k\le mn.
\end{equation}
Then
\begin{equation}\label{lex2}
\sg_{\rm st}=\left\{\langle x_{2l-1},x_{2l}\rangle,\; l=1,\ldots,\frac{mn}{2}\right\}.
\end{equation}
Observe that
\begin{equation}\label{lex3}
\sgn(\sg_{\rm st};O _i)=+1,\quad i=1,2,3,4,
\end{equation}
because $\pi(\sg_{\rm st})={\rm Id}$ and $\ep_i( x_{2l-1},x_{2l})=+1$.

We will use the following lemma: 

\begin{lem} (see \cite{Kenyon}, \cite{Tes}). \label{Co}
Let $\sg,\sg'$ be any two configurations and $\{\ga_k\ :\ k=1,2,\ldots,r\}$ all contours of $\sg\cup\sg'.$ Then
\begin{equation}\label{Imp_f}
\sgn(\sg;O_i)\cdot\sgn(\sg';O_i)=\prod_{k=1}^r\sgn(\gamma_k;O_i),\quad i=1,2,3,4,
\end{equation} 
with 
\begin{equation}\label{Imp_fp}
\sgn(\gamma_k;O_i)=(-1)^{\nu_k(O_i)+1},
\end{equation} 
where $\nu_k(O_i)$ is the number of edges in $\gamma_k$ oriented clockwise with respect to the orientation $O_i$. 
\end{lem}

For the convenience of the reader we give a proof of this lemma in Appendix \ref{appA}.

\section{Preliminary results} \label{preliminary}

As shown by Kasteleyn \cite{Kas1,Kas3} (for recent expositions see the works \cite{Kenyon}, \cite{McCoy}, \cite{McCoyWu} and references therein), the partition function $Z$ can be decomposed as
\begin{equation}\label{a2}
Z=Z^{00}+Z^{10}+Z^{01}+Z^{11},
\end{equation}
the four partition functions $Z^{rs}$ corresponding to dimer configurations of the homology classes $(r,s)\in \Z_2\oplus\Z_2$ with respect to the standard configuration $\sg_{\rm st}$, 
and the Pfaffians $\Pf A_i$ are expressed as
\begin{equation}\label{a3}
\begin{aligned}
& \Pf A_1=Z^{00}-Z^{10}-Z^{01}-Z^{11}, \quad \Pf A_2=Z^{00}-Z^{10}+Z^{01}+Z^{11}, \\ 
& \Pf A_3=Z^{00}+Z^{10}-Z^{01}+Z^{11}, \quad \Pf A_4=Z^{00}+Z^{10}+Z^{01}-Z^{11}.
\end{aligned}
\end{equation}
These equations are called the {\it Kasteleyn identities}.
Observe that equations \eqref{a2}, \eqref{a3} imply \eqref{a1}.

The proof of the Kasteleyn identities \eqref{a3} can be found in the works of Galluccio and Loebl \cite{GalluLoe}, Tesler \cite{Tes}, and Cimasoni and Reshetikhin \cite{CimResh}. It follows from formula \eqref{PfOurDef} 
that the Pfaffians $\Pf A_i$ are multivariate polynomials with respect to the weights $z_h,z_v,z_d$. By diagonalizing the matrices $A_i$, one can obtain the double product formulas, 
\begin{equation}\label{a5}
\begin{aligned}
\det A_i=\prod_{j=0}^{\frac{m}{2}-1}\prod_{k=0}^{n-1}4\Bigg[&z_h^2\sin^2\frac{2\pi(j+\alpha_i)}{m}+z_v^2\sin^2\frac{2\pi(k+\beta_i)}{n}\\
&+z_d^2\cos^2\left(\frac{2\pi(j+\alpha_i)}{m}+\frac{2\pi(k+\beta_i)}{n}\right)\Bigg],
\end{aligned}
\end{equation}
with
\begin{equation}\label{a6}
\begin{aligned}
\al_1=\be_1=0\,;\quad \al_2=0\,,\;\be_2=\frac{1}{2}\,;\quad\al_3=\frac{1}{2}\,,\;\be_3=0\,;\quad\al_4=\be_4=\frac{1}{2}\,;
\end{aligned}
\end{equation}
(see, e.g., \cite{Fend} and  \cite{IzmaKenna}).

The function 
\begin{equation}\label{sf}
\begin{aligned}
S(x,y)=4\left[ z_h^2\sin^2 2\pi x+z_v^2\sin^2 2\pi y
+z_d^2\cos^2\left(2\pi x+2\pi y\right)\right]
\end{aligned}
\end{equation}
is the {\it spectral function} of the dimer model on the triangular lattice.
In its terms equation \eqref{a5} is conveniently written as
\begin{equation}\label{sf1}
\begin{aligned}
\det A_i=\prod_{j=0}^{\frac{m}{2}-1}\prod_{k=0}^{n-1} S\left( \frac{j+\al_i}{m},
\frac{k+\be_i}{n}\right).
\end{aligned}
\end{equation}
The function $S(x,y)$ is periodic in $x$ and $y$,
\begin{equation}\label{sf2}
S\left(x+\frac{1}{2},y\right)=S\left(x,y+\frac{1}{2}\right)=S(x,y),
\end{equation}
and if $z_h,z_v,z_d>0$, then 
\begin{equation}\label{sf3}
S(x,y)>0,\quad \forall (x,y)\in \R^2.
\end{equation}
Indeed, obviously, $S(x,y)\ge 0$.
Suppose $S(x,y)=0$. Then from \eqref{sf} we obtain that
\begin{equation}\label{sf4}
2x\in\Z, \quad 2y\in \Z, \quad 2x+2y\in\frac{1}{2}+ \Z,  
\end{equation}
which are inconsistent. 
From \eqref{sf1}, \eqref{sf3} we obtain that
\begin{equation}\label{a7}
\begin{aligned}
\det A_i>0,\quad \textrm{if}\quad z_h,z_v,z_d>0,
\end{aligned}
\end{equation}
because all the factors in \eqref{sf1} are positive.

As a consequence of \eqref{a7}, we have that $\Pf A_i$ does not change the sign in the region
$z_h,z_v,z_d>0$; hence, it is sufficient to establish the sign of $\Pf A_i$ at any point of the region $z_h,z_v,z_d>0$. As a first step in this direction,
let us prove the following
proposition:

\begin{prop} \label {detA} We have that
\begin{enumerate}
  \item If $z_h>0$, $z_v>0$, and $z_d=0$, then 
\begin{equation}\label{a9}
\begin{aligned}
&\det A_1=0,\quad  \det A_3>0, \quad \det A_4>0,\\
&\det A_2
\left\{
\begin{aligned}
&=0,\quad \textrm{if $n\equiv 1\pmod 2$},\\
&>0,\quad \textrm{if $n\equiv 0\pmod 2$}.
\end{aligned}
\right.
\end{aligned}
\end{equation}
  \item If $z_h>0$, $z_v=0$, and $z_d>0$, then 
\begin{equation}\label{a11}
\begin{aligned}
&\det A_1
\left\{
\begin{aligned}
&=0,\quad \textrm{if $n\equiv 0 \pmod 4$},\\
&>0,\quad \textrm{if $n\not\equiv 0\pmod 4$}.
\end{aligned}
\right.\\
&\det A_2
\left\{
\begin{aligned}
&=0,\quad \textrm{if $n\equiv 2 \pmod 4$},\\
&>0,\quad \textrm{if $n\not\equiv 2 \pmod 4$}.
\end{aligned}
\right.\\
&\det A_3>0, \quad \det A_4>0.
\end{aligned}
\end{equation}
  \item If $z_h=0$, $z_v>0$, and $z_d>0$, then 
\begin{equation}\label{a12}
\begin{aligned}
&\det A_1
\left\{
\begin{aligned}
&=0,\quad \textrm{if $m\equiv 0 \pmod 4$},\\
&>0,\quad \textrm{if $m\equiv 2 \pmod 4$}.
\end{aligned}
\right.\\
&\det A_2
\left\{
\begin{aligned}
&=0,\quad \textrm{if $m\equiv 0 \pmod 4$ and $n\equiv 1 \pmod 2$},\\
&>0,\quad \textrm{otherwise}.
\end{aligned}
\right.\\
&\det A_3
\left\{
\begin{aligned}
&=0,\quad \textrm{if $m\equiv 2 \pmod 4$},\\
&>0,\quad \textrm{if $m\equiv 0 \pmod 4$}.
\end{aligned}
\right.\\
&\det A_4
\left\{
\begin{aligned}
&=0,\quad \textrm{if $m\equiv 2 \pmod 4$ and $n\equiv 1 \pmod 2$},\\
&>0,\quad \textrm{otherwise}.
\end{aligned}
\right.
\end{aligned}
\end{equation}
\end{enumerate}
\end{prop}

\begin{proof} (1) Assume that $z_h>0$, $z_v>0$, and $z_d=0$. By \eqref{a5}, $\det A_i=0$ if and only if for some 
$0\le j\le \frac{m}{2}-1$ and $0\le k\le n-1$,
\begin{equation}\label{aa1}
\begin{aligned}
\frac{2(j+\alpha_i)}{m}\in\Z\,\quad {\rm and} \quad \frac{2(k+\beta_i)}{n}\in\Z.
\end{aligned}
\end{equation}
In particular, this gives that 
$\det A_1=0$, due to the factor $j=k=0$.
On the other hand, $\det A_3>0$, because $\frac{2j+1}{m} \not\in\Z$
(since $m$ is even).
The same argument works for $\det A_4>0$.

Consider $\det A_2$.
If $n\equiv 1\pmod 2$, then $\det A_2=0$, due to the factor $j=0$, $k=\frac{n-1}{2}$. On the other hand,
if $n\equiv 0\pmod 2$, then 
$
\frac{2k+1}{n}\not\in\Z
$
and $\det A_2\not=0$.

(2) Assume that $z_h>0$, $z_v=0$, and $z_d>0$. To have $\det A_i=0$, we need that
\begin{equation}\label{aa2}
\begin{aligned}
\frac{2(j+\alpha_i)}{m}\in\Z\,\quad {\rm and} \quad \frac{2(k+\beta_i)}{n}
\in \frac{1}{2}+\Z.
\end{aligned}
\end{equation}
We have that $\det A_1=0$, provided $n\equiv 0 \pmod 4$,
due to the factor $j=0$, $k=\frac{n}{4}\,$. On the other hand, if $n\not \equiv 0 \pmod 4$, then 
$\frac{2 k}{n}\not \in\frac{1}{2}+\Z$, hence $\det A_1>0$.

Consider $\det A_2$.
If $n\equiv 2 \pmod 4$, then $\det A_2=0,$ due to the factor $j=0$, $k=\frac{n-2}{4}.$ 
On the other hand, if $n\not\equiv 2 \pmod 4$, then $\frac{2k+1}{n}\not\in\frac{1}{2}+\Z$, 
and therefore $\det A_2>0$. Also, $\det A_3>0$, because $\frac{2j+1}{m}\not \in\Z$.
The same argument is applied to $\det A_4$.

(3) Assume that $z_h=0$, $z_v>0$, and $z_d>0$.  We need that
\begin{equation}\label{aa3}
\begin{aligned}
\frac{2(j+\alpha_i)}{m}\in\frac{1}{2}+\Z\,\quad {\rm and} \quad 
\frac{2(k+\beta_i)}{n}\in\Z.
\end{aligned}
\end{equation}
This equation is similar to \eqref{aa2}, but since we assume that $m$ is even and $n$ can be odd,
the analysis is slightly different.
If $m\equiv 0 \pmod 4$, then $\det A_1=0,$ due to the factor $j=\frac{m}{4}$, $k=0.$ 
On the other hand, if $m\equiv 2 \pmod 4$, then $\frac{2 j}{m}\not\in\frac{1}{2}+\Z$, hence $\det A_1>0$. 

If $m\equiv 0 \pmod 4$ and $n\equiv 1 \pmod 2$, then $\det A_2=0$, due to the factor $j=\frac{m}{4}$, $k=\frac{n-1}{2}$. On the other hand,
if $n\equiv 0 \pmod 2$, then 
$\frac{2k+1}{n}\not \in\Z$, hence $\det A_2>0$. If $m\not\equiv 0 \pmod 4$, 
then $\frac{2j}{m}\not\in \frac{1}{2} +\Z$, 
hence $\det A_2>0.$

If $m\equiv 2 \pmod 4$, then $\det A_3=0$, due to the factor $j=\frac{m-2}{4}$, $k=0$. On the other hand, 
if  $m\equiv 0 \pmod 4$,
then $\frac{2j+1}{m}\not\in \frac{1}{2} +\Z$, hence $\det A_3>0$.

If $m\equiv 2 \pmod 4$ and $n\equiv 1 \pmod 2$, then $\det A_4=0$, due to the factor $j=\frac{m-2}{4}$, $k=\frac{n-1}{2}$. 
On the other hand, if $n\equiv 0 \pmod 2$, then $\frac{2k+1}{n}\not\in\Z$, hence $\det A_4>0$.
If $m\equiv 0 \pmod 4$, then $\frac{2j+1}{m}\not\in \frac{1}{2} +\Z$, hence $\det A_4>0$. 
\end{proof}

\section{Positivity of $\Pf A_3$ and $\Pf A_4$} \label{A34}

Let us turn to the proof of Theorem \ref{PST}.
We first prove that
 $\Pf A_3>0$ and $\Pf A_4>0$.

\vskip 2mm

\begin{lem}
Let $z_h,z_v,z_d>0$. Then 
\begin{equation}
\Pf A_3>0,\quad \Pf A_4>0.
\end{equation}
\end{lem}
\begin{proof}
If $z_h>0$, $z_v>0$, and $z_d=0$, then by Proposition \ref{detA} (1), 
$\det A_1=0$, hence from \eqref{a3} we deduce that
\begin{equation}\label{a13}
Z^{00}-Z^{01}=Z^{10}+Z^{11}. 
\end{equation}
On the other hand, if $z_h>0$, $z_v>0$, and $z_d=0$, then $\det A_3>0$. This implies that $\Pf A_3\not=0$, hence
from \eqref{a3}, \eqref{a9}, and nonnegativity of $Z^{rs}$ we obtain that
\begin{equation}\label{a14}
\begin{aligned}
\Pf A_3= 2Z^{10}+2Z^{11}>0,\quad \textrm{if}\quad z_h>0, \quad z_v>0,\quad z_d=0.
\end{aligned}
\end{equation}
By continuity, $\Pf A_3>0$ for the chosen $z_h>0$, $z_v>0$, and small $z_d>0$.
This proves that $\Pf A_3>0$ in the whole region
$z_h,z_v,z_d>0$. The same argument works for $\Pf A_4>0$.
\end{proof}
We will finish the proof of Theorem \ref{PST} in the subsequent two sections by showing that $\Pf A_2>0$ and $\Pf A_1<0$, respectively.

\section{Identities $\Pf A_1=-\Pf A_2$, $\Pf A_3=\Pf A_4$ for odd $n$} \label{Id}

\begin{lem}\label{bijection}
Let $n\equiv 1\pmod 2$. Then $Z^{00}=Z^{10}$ and $Z^{01}=Z^{11}$.
\end{lem}

\begin{proof}
Recall that $Z^{rs}$ is the partition function corresponding to dimer configurations in homology class $(r,s)\in \Z_2\oplus\Z_2$ with respect to the standard configuration $\sg_{st}$. The homology class of a given configuration $\sg$ can be calculated as follows:
let $v_j$ be the number of intersections of the vertical line $x=j+\frac{1}{2}$ with the dimers in $\sg\cup\sg_{st}$. Then 
\begin{equation}\label{rbij0}
v_j\equiv r\pmod 2 \quad \forall\,j\in\Z_m.
\end{equation}
Obviously,
\begin{equation}\label{rbij}
 v_j= v_j(\sg)+v_j(\sg_{st})\,,
\end{equation}
where $v_j(\sg)$ denotes the number of intersections of the vertical line $x=j+\frac{1}{2}$ with the dimers in $\sg$. Let $T\sg$ be a shift of the configuration $\sg$ to the right by 1. Then 
\begin{equation}\label{rbij1}
\begin{aligned}
v_{j+1}(T\sg)=v_j(\sg).
\end{aligned}
\end{equation}
For the standard configuration the number of intersections is
\begin{equation}\label{rbij2}
v_j(\sg_{\rm st})=
\left\{
\begin{aligned}
&n,\quad \textrm{if}\quad j\equiv 0\pmod 2\,,\\
&0,\quad \textrm{if}\quad j\equiv 1\pmod 2\,,
\end{aligned}
\right.
\end{equation}
hence  the number of intersections of the vertical line $x=j+1+\frac{1}{2}$ with the dimers in $T\sg\cup\sg_{st}$ is equal to  
\begin{equation}\label{rbij3}
  v_{j+1}(T\sg)+v_{j+1}(\sg_{st})=v_j(\sg)+v_j(\sg_{\rm st})+(-1)^{j+1}n\,.
\end{equation}
This implies the relation between the homology class numbers $r(T\sg)$ and $r(\sg)$ as
\begin{equation}\label{rbij4}
 r(T\sg)=r(\sg)+n \pmod 2.
\end{equation}
In particular, if $n$ is odd then
\begin{equation}\label{rbij5}
 r(T\sg)=r(\sg)+1 \pmod 2.
\end{equation}

Similarly, let $h_k$ be the number of intersections of the horizontal line $y=k+\frac{1}{2}$ with the dimers in $\sg\cup\sg_{st}$. Then 
\begin{equation}\label{rbij6}
h_k\equiv s\pmod 2 \quad \forall\,k\in\Z_n
\end{equation}
and
\begin{equation}\label{rbij7}
 h_k= h_k(\sg)+h_k(\sg_{st})\,,
\end{equation}
where $h_k(\sg)$ denotes the number of intersections of the line $y=k+\frac{1}{2}$ with the dimers in $\sg$. Obviously, $h_k(\sg_{st})=0$, hence
\begin{equation}\label{rbij8}
 h_k= h_k(\sg)\,.
\end{equation}
Also,
\begin{equation}\label{rbij9}
  h_k(T\sg)=h_k(\sg)\,,
\end{equation}
hence
\begin{equation}\label{rbij10}
  s(T\sg)=s(\sg)\,.
\end{equation}
Combining this with \eqref{rbij5}, we obtain that if $n$ is odd, then 
\begin{equation}\label{rbij11}
T\,:\, \Sigma_{m,n}^{0,0}\to \Sigma_{m,n}^{1,0}\,,\quad
T\,:\, \Sigma_{m,n}^{0,1}\to \Sigma_{m,n}^{1,1}\,,
\end{equation}
where $\Sigma_{m,n}^{r,s}$ is the set of dimer configurations in
the homology class $(r,s)$. Since the shift $T$ is invertible, the mappings
\eqref{rbij11} are bijections.

Since $w(T\sg)=w(\sg)$ we obtain that $Z^{00}=Z^{10}$ and $Z^{01}=Z^{11}$.
\end{proof}

Note that from Lemma \ref{bijection} and the Kasteleyn identities \eqref{a3}, we immediately have the following theorem:
\begin{theo}\label{Conj} Let $n\equiv 1\pmod 2$. Then 
for all $z_h,z_v,z_d\ge 0$,
\begin{equation}\label{con0}
\Pf A_1=-\Pf A_2,\quad \Pf A_3=\Pf A_4.
\end{equation}
\end{theo}

Numeric data for Pfaffians $\Pf A_i$ in Appendix \ref{numerics} illustrate the identities
$\Pf A_1=-\Pf A_2$ and $\Pf A_3=\Pf A_4$ on the $4\times 3$ lattice, while they disprove
such identities on the $4\times 4$ lattice. For some interesting identities between dimer model partition functions on different surfaces see the paper of Cimasoni and Pham \cite{CimPham}.

\section{Positivity of $\Pf A_2$} \label{A2}

We begin with the following case:

\begin{lem}
Let $z_h,z_v,z_d>0$. Then if either $n\equiv 0\pmod 2$ or $m\equiv 2\pmod 4$, then 
\begin{equation}\label{A2pos1}
\Pf A_2>0.
\end{equation}
\end{lem}
\begin{proof}
First assume that $n\equiv 0\pmod 2$ and consider the case $z_h>0$, $z_v>0$, and $z_d=0$. 
Since $n$ is even, then by Proposition \ref{detA} (1),
$\det A_2>0$ and $\det A_1=0$, hence 
\begin{equation} Z^{00}-Z^{10}=Z^{01}+Z^{11}\implies \Pf A_2=2Z^{01}+2Z^{11}>0. 
\end{equation} 
By continuity, $\Pf A_2>0$ for the chosen $z_h$, $z_v$, and small $z_d$. This proves $\Pf A_2>0$ in the whole region $z_h,z_v,z_d>0$. 
Now assume that $m\equiv 2\pmod 4$ and consider the case $z_h=0$, $z_v>0$, and $z_d>0$. By Proposition \ref{detA} (3),
in this case $\det A_2>0$, $\det A_3=0$, hence
\begin{equation}\label{a15}
Z^{01}-Z^{10}=Z^{00}+Z^{11} \implies  \Pf A_2= 2Z^{00}+2Z^{11}>0.
\end{equation}
By continuity, $\Pf A_2>0$ for the chosen $z_v>0$, $z_d>0$, and small $z_h>0$.
This proves that $\Pf A_2>0$ in the whole region
$z_h,z_v,z_d>0$.
\end{proof}

The positivity of $\Pf A_2$ in the 
case $m\equiv 0\pmod 4$, $n\equiv 1\pmod 2$ is more difficult. 
To deal with this case,  we consider the asymptotic behavior of 
$\Pf A_2$ for $z_h=1$, $z_v=0$, as $z_d\to 0$, and prove the following lemma:

\begin{lem} \label{lem1} Let $m\equiv 0\pmod 4$, $n\equiv 1\pmod 2$,  and  $z_h=1$, $z_v=0$. Then as $z_d\to 0$,
\begin{equation}\label{pos1}
\Pf A_2=2\left(\frac{m}{2}\right)^n  z_d^n(1+\mathcal O(z_d)).
\end{equation} 
\end{lem}

{\it Remark.} This will imply that $\Pf A_2>0$ for $z_h=1$, $z_v=0$ and sufficiently small $z_d$, and hence $\Pf A_2>0$ for all
$z_h,z_v,z_d>0$.  

\begin{proof} We have that 
\begin{equation}\label{pos2}
\Pf A_2=\sum_{\sg\in\Sigma_{m,n}} \sgn(\sg) w(\sg).
\end{equation} 
By our assumption, $z_v=0$, hence there are no vertical dimers.
Consider first the limiting case,  $z_d=0$. In this case there are only horizontal dimers.
Let $\sg$ be any configuration of horizontal dimers and $T_k\sg$ a configuration obtained from $\sg$ by the shift $x\to x+\mathsf e_1$,
$\mathsf e_1=(1,0)$, on the horizontal line $y=k$. Then the superposition
$\sg\cup T_k\sg$ consists of trivial contours and one horizontal contour around the torus, which has a negative sign as there are $m\equiv 0 \pmod 2$ arrows in the direction of movement from left to right. Hence,
$\sgn(T_k\sg)=-\sgn(\sg)$, $w(T_k\sg)=w(\sg)$. From here it follows that
\[
\sgn(\sg) w(\sg)+\sgn(T_k\sg)w(T_k\sg)=0,
\]
and therefore, $\sg$ and $T_k(\sg)$ cancel each other in $\Pf A_2$. This implies that
\begin{equation}\label{pos2.1}
\Pf A_2\big|_{z_d=0}=0.
\end{equation} 

Now let $z_d>0$, so that there are both horizontal and diagonal dimers. We will consider $z_d\to 0$,
and we will call configurations consisting of only horizontal dimers the {\it ground state}
configurations, because they have the biggest weight $w(\sg)$. As we saw, the ground state configurations
cancel each other in $\Pf A_2$. We will call not ground state configurations  {\it excited state} configurations.

Consider the set of configurations $\Sigma(k)$
in which the line $y=k$ is occupied completely by horizontal dimers. Let $\sg\in \Sigma(k)$ and $T_k\sg$ obtained
from $\sg$ by the shift  $x\to x+\mathsf e_1$ on the  line $y=k$. Then  
again $\sgn(T_k\sg)=-\sgn(\sg)$ and hence
\[
\sum_{\sg\in\Sigma(k)} \sgn(\sg) w(\sg)=0.
\]
By adding over $k\in\Z_n$, we obtain that
\begin{equation}\label{pos3}
\sum_{\sg\in\bigcup_{k\in\Z_n}\Sigma(k)} \sgn(\sg) w(\sg)=0,
\end{equation}
hence to get excited states without cancellations, we have to consider the set of configurations  such that
on each line $y=k$ there is at least one vertex covered by a diagonal dimer. In fact, since $m$ is even, there are  
at least two vertices covered by diagonal dimers, 
hence the total number of diagonal dimers $ N_{\rm diag}(\sg)$
is at least $n$.
We denote 
\begin{equation}\label{Sg0}
\Sigma_0=\{\sg\in\Sigma_{m,n}\,|\,N_{\rm diag}(\sg,k)=2,\;k\in\Z_n\},
\end{equation}
where $N_{\rm diag}(\sg,k)$ is the number of vertices on the line $y=k$ in $\sg$ covered by diagonal dimers.

Denote by $N_{\rm diag}(\sg,k,k+1)$ the number of diagonal dimers connecting horizontal line $y=k$ to horizontal line $y=k+1$. Then 
\[
N_{\rm diag}(\sg,k)=N_{\rm diag}(\sg,k-1,k)+N_{\rm diag}(\sg,k,k+1),
\]
hence
\[
N_{\rm diag}(\sg,k,k+1)=0,\, 1 \;\textrm{or}\; 2,\quad \forall \sg\in\Sg_0. 
\]
Suppose that for some $\sg\in\Sg_0$ and some $k\in \Z_n$,
\[
N_{\rm diag}(\sg,k,k+1)= 0, 
\]
then
\[
N_{\rm diag}(\sg,k+1,k+2)= 2, \quad N_{\rm diag}(\sg,k+2,k+3)= 0, \quad
N_{\rm diag}(\sg,k+3,k+4)= 2, \ldots,
\]
hence the whole lattice is stratified into horizontal strips of width 2 with 2 diagonal dimers in each strip. But $n$ is odd, hence such stratification is not possible. Similarly, $N_{\rm diag}(\sigma,k,k+1) = 2$ is not possible as well. This implies that if $\sg\in\Sg_0$, then
\[
N_{\rm diag}(\sg,k,k+1)= 1,\quad \forall k\in\Z_n.  
\]
This means that for every $k$ there is a unique diagonal dimer connecting horizontal lines $y=k$ and $y=k+1$. Let $(j_k,k)$ be the vertex covered by this diagonal dimer
on the horizontal line $y=k$.
Then, since all vertices between $(j_{k}+1,k+1)$  and $(j_{k+1},k+1)$ must be covered by horizontal dimers we have that 
\begin{equation}\label{pos4a}
j_{k+1}-(j_k+1)\equiv 1\pmod 2, \;\forall k\in\Z_n,
\end{equation}
or, equivalently,
\begin{equation}\label{pos4}
j_{k+1}-j_k\equiv 0\pmod 2, \;\,\forall k\in\Z_n.
\end{equation}
This implies that 
\begin{equation}\label{pos4ab}
j_k-j_{\ell}\equiv 0 \pmod 2, \;\,\forall k,\ell\in\Z_n.
\end{equation}
Hence, 
\begin{equation}\label{pos4b}
\Sigma_{0}=\bigsqcup_{i=0}^{1}\Sigma_{0}^{i}
\end{equation}
and $\sg\in\Sigma_{0}^{i},\; i=0,1,$ if
\begin{equation}\label{pos4c}
j_\ell\equiv i \pmod 2\;\,\forall \{(j_\ell,k_\ell),(j_\ell+1, k_\ell+1)\}.
\end{equation}
The proof of Lemma \ref{lem1} is based on the following lemma:

\begin{lem} \label{lem2} For any $\sg\in\Sigma_0$,
\begin{equation}\label{pos5}
\sgn(\sg)=+1.
\end{equation} 
\end{lem}

Let us finish the proof of Lemma \ref{lem1}, assuming Lemma \ref{lem2}.
Since $w(\sg)=z_d^n$ (recall that $z_h=1$) for any $\sg\in \Sigma_0$, we obtain that
\begin{equation}\label{pos6}
\sum_{\sg\in\Sigma_0} \sgn(\sg) w(\sg)=|\Sigma_0|z_d^n.
\end{equation} 
At $y=0$ we have $m$ choices for $j_0$ and then,
because of condition \eqref{pos4}, we have $\frac{m}{2}$ choices for $j_1,\,j_2,\,\ldots,\, j_{n-1}$. Therefore,
\begin{equation}\label{pos7}
|\Sigma_0|=m\left(\frac{m}{2}\right)^{n-1},
\end{equation} 
hence
\begin{equation}\label{pos8}
\sum_{\sg\in\Sigma_0} \sgn(\sg) w(\sg)=m\left(\frac{m}{2}\right)^{n-1}z_d^n.
\end{equation} 
The higher excited states have at least $(n+1)$ diagonal dimers and therefore their weight is at most $z_d^{n+1}$, hence Lemma \ref{lem1} follows. \end{proof}

It remains to prove Lemma \ref{lem2}. We define a \textit{stack configuration}, $\sg_{\rm stack},$ to be a configuration in which all diagonal dimers form a stack between the vertical lines $x=0$ and $x=1$. The remaining vertices are occupied by horizontal dimers. See Fig.\ \ref{FPfA2} (a).

\noindent {\it Proof of Lemma \ref{lem2}}. The proof consists of two steps. At Step 1 we show that if $\sigma\in\Sigma_0^0$ (see equations \eqref{pos4b},  \eqref{pos4c}), then 
\begin{equation}\label{pos9}
\sgn(\sigma)=\sgn(\sigma_{\rm stack}),
\end{equation}
and if $\sg\in\Sigma_0^1,$ then
\begin{equation}\label{pos9a}
\sgn(\sigma)=\sgn(\sigma_{\rm stack}+\mathsf e_1).
\end{equation}
At Step 2 we show that 
\begin{equation}\label{pos10}
\sgn(\sigma_{\rm stack})=\sgn(\sigma_{\rm stack}+\mathsf e_1)=+1,
\end{equation}
hence $\sgn(\sigma)=+1$, $\forall\sg\in\Sigma_0$.\\

{\it Step 1.} Observe that any configuration $\sg\in\Sigma_{0}$ is determined
by the position of its diagonal dimers. Let us call an {\it elementary move} the change of $\sg$ to $\sg'$, 
where a diagonal dimer $\{(j,k),(j+1,k+1)\}$ is shifted
 to $\{(j+2,k),(j+3,k+1)\}$. Assume that in $\sg$ the intermediate vertices $(j+1,k)$ and $(j+2,k+1)$
are not covered by diagonal dimers.
Then the superposition $\sg\cup\sg'$ consists of trivial contours and one nontrivial
contour of the length 6, which is positive as shown in Fig.\ \ref{F1} (a) (note that there are always three arrows opposite to any direction of movement along that contour), hence $\sgn(\sg)=\sgn(\sg')$. If exactly one of the vertices $(j+1,k)$, 
 $(j+2,k+1)$ is covered by a diagonal dimer, then the superposition $\sg\cup\sg'$ consists of trivial contours and one nontrivial
contour of the length $m+2$, $m+1\equiv 1\pmod 2$ of whose arrows are oriented from left to right, as shown in Fig.\ \ref{F1} (b), hence again $\sgn(\sg)=\sgn(\sg')$. Finally, if both 
vertices $(j+1,k)$, 
$(j+2,k+1)$ are covered by diagonal dimers, then $\sgn(\sg)=\sgn(\sg')$, because we can consider
a clockwise sequence of elementary moves from $\sg$ to $\sg'$, without intermediate vertices covered by diagonal dimers. See Fig.\ \ref{F1} (c). Thus, for any elementary move $\sg\to\sg'$ we have that
\begin{equation}\label{sgn-elem}
\sgn(\sigma')=\sgn(\sigma).
\end{equation}

If $\sg\in\Sigma_{0}^{0},$ then by elementary moves we can first move the diagonal dimer at the horizontal line $y=0$ to the position  $\{(0,0), (1,1)\},$ and then inductively at horizontal line $y=k,$ $k = 1, 2, \ldots, n-1,$ to the position  $\{(0,k), (1,k+1)\},$ forming a stack of diagonal dimers above the dimer at the horizontal line $x=0.$ In other words, we have moved $\sg$ to $\sg_{\rm stack}.$ If $\sg\in\Sigma_{0}^{1},$ then using the above argument we move it to $\sg_ {\rm stack}+\mathsf e_1.$ Hence, formulae \eqref{pos9} and \eqref{pos9a} hold.\\

{\it Step 2.} As shown in Fig.\ \ref{FPfA2} (b),
the superposition $\sigma_{\rm stack}\cup\sigma_{\rm st}$ consists of
trivial contours and exactly one nontrivial contour which is a zigzag path between the vertical lines $x=0$ and $x=1$. The latter one is positive, since there are $2n-1\equiv 1\pmod 2$ arrows in the direction of movement from top to bottom and hence $\sgn(\sg_{\rm stack})=+1.$ 

Similarly, the superposition $(\sigma_{\rm stack}+\mathsf e_1)\cup(\sigma_ {\rm st}+\mathsf e_1)$  consists of trivial contours and a zigzag contour with $n+1\equiv 0\pmod 2$ arrows in the direction of movement from top to bottom, and so $\sgn((\sigma_{\rm stack}+\mathsf e_1)\cup(\sigma_{\rm st}+\mathsf e_1))=-1.$ In addition, the sign of the superposition $\sigma_{\rm st}\cup(\sigma_{\rm st}+\mathsf e_1)$  is $(-1)$ as well. Indeed, the superposition $\sigma_{\rm st}\cup(\sigma_{\rm st}+\mathsf e_1)$  consists of 
$n$ horizontal contours of the length $m$. Since $m$ is even, the sign of 
each contour is $(-1)$, and
since $n$ is odd, the sign of the superposition $\sigma_{\rm st}\cup(\sigma_{\rm st}+\mathsf e_1)$  is $(-1)$ as well. Now, since the sign of the superpositions 
$(\sigma_{\rm stack}+\mathsf e_1)\cup(\sigma_{\rm st}+\mathsf e_1)$ and 
$\sigma_{\rm st}\cup(\sigma_{\rm st}+\mathsf e_1)$ is $(-1)$, we obtain that
$\sgn((\sigma_{\rm stack}+\mathsf e_1)\cup \sigma_{\rm st})=+1$,
hence $\sgn(\sigma_{\rm stack}+\mathsf e_1)=+1.$

Lemma \ref{lem2} is proven.$\hfill\square$
\begin{figure}[h]
\begin{tabular}{c c c}
\includegraphics[scale=.35]{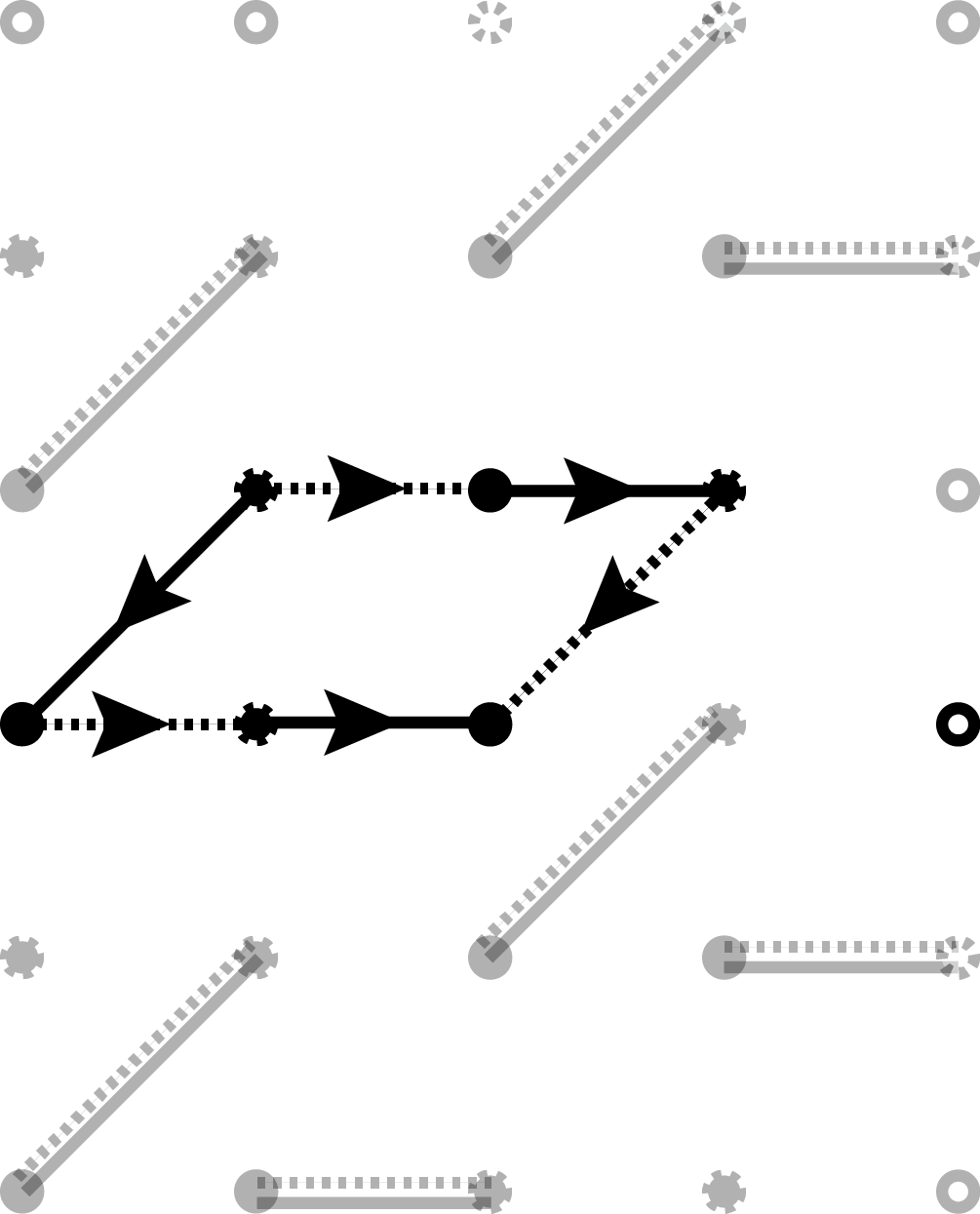}&\hspace{.33in} \includegraphics[scale=.35]{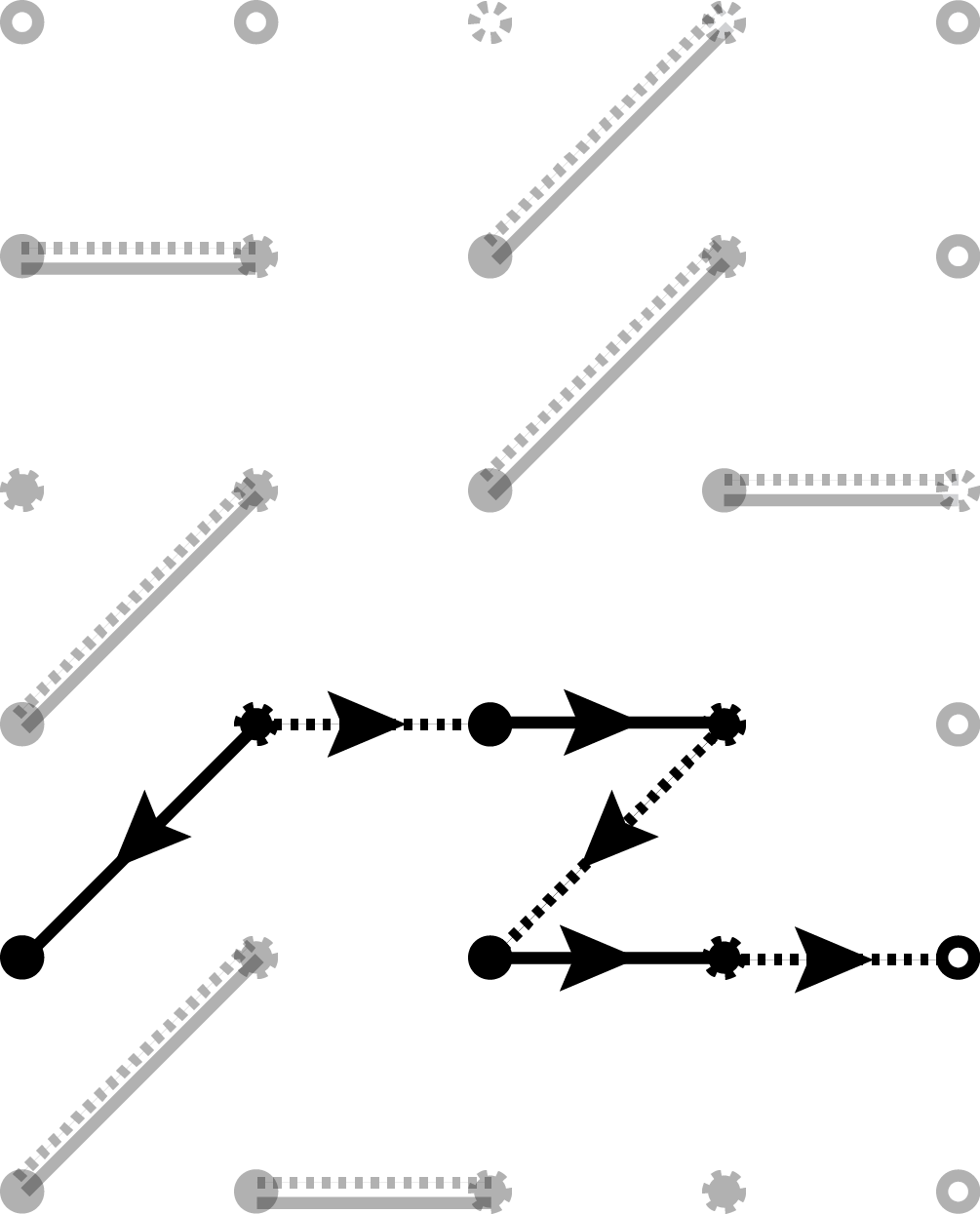}&\hspace{.33in} \includegraphics[scale=.35]{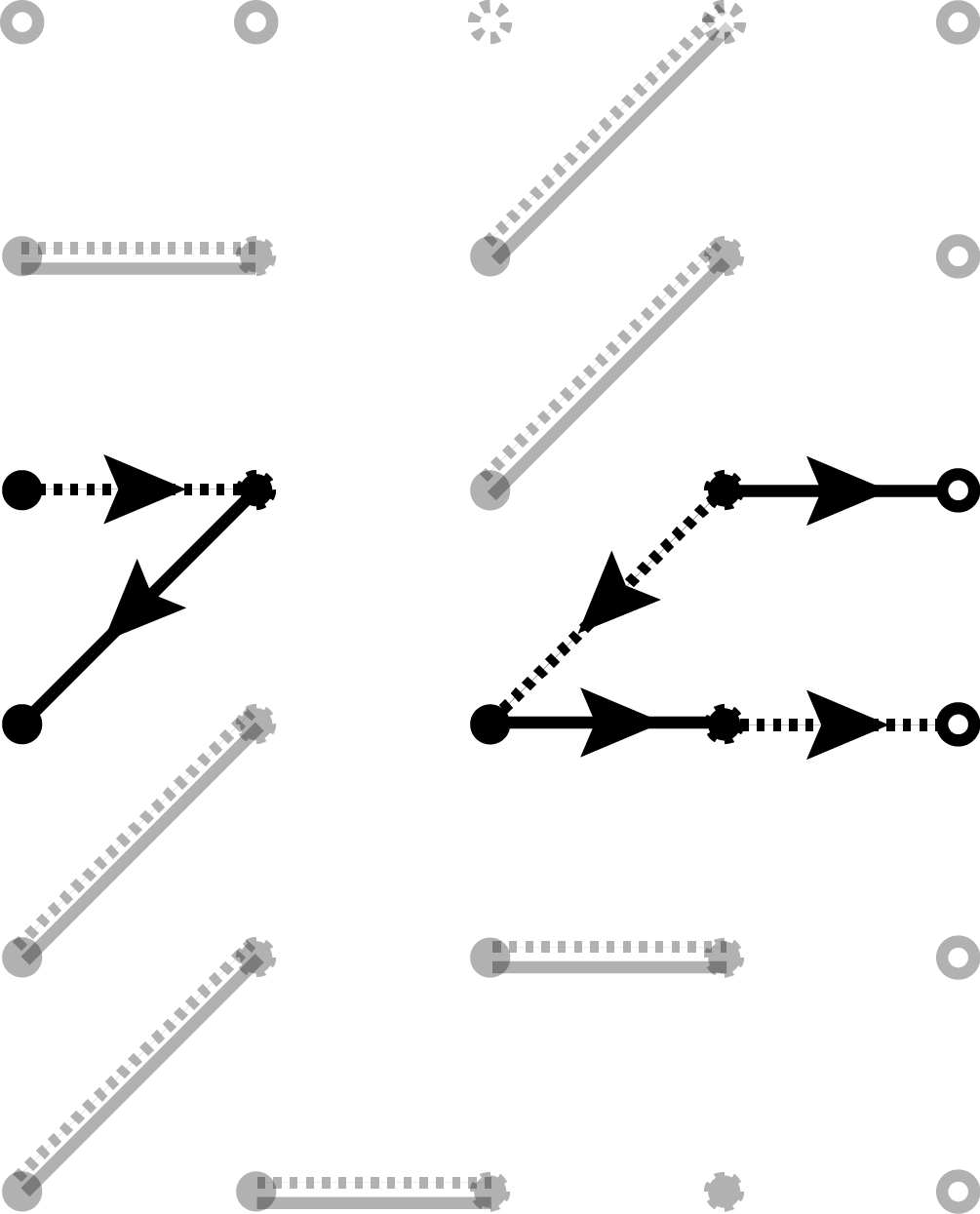}\\ ${}$ \\
(a) $\sg\cup\sg'$ &\hspace{.2in} (b) $\sg\cup\sg'$ &\hspace{.2in} (c) $\sg\cup\sg'$
\end{tabular}
\caption{Examples of different types of the superpositions $\sg\cup\sg'$ on a $4\times 5$ lattice.}
\label{F1}
\end{figure}

\begin{figure}[h]
\begin{tabular}{c c}
\includegraphics[scale=.35]{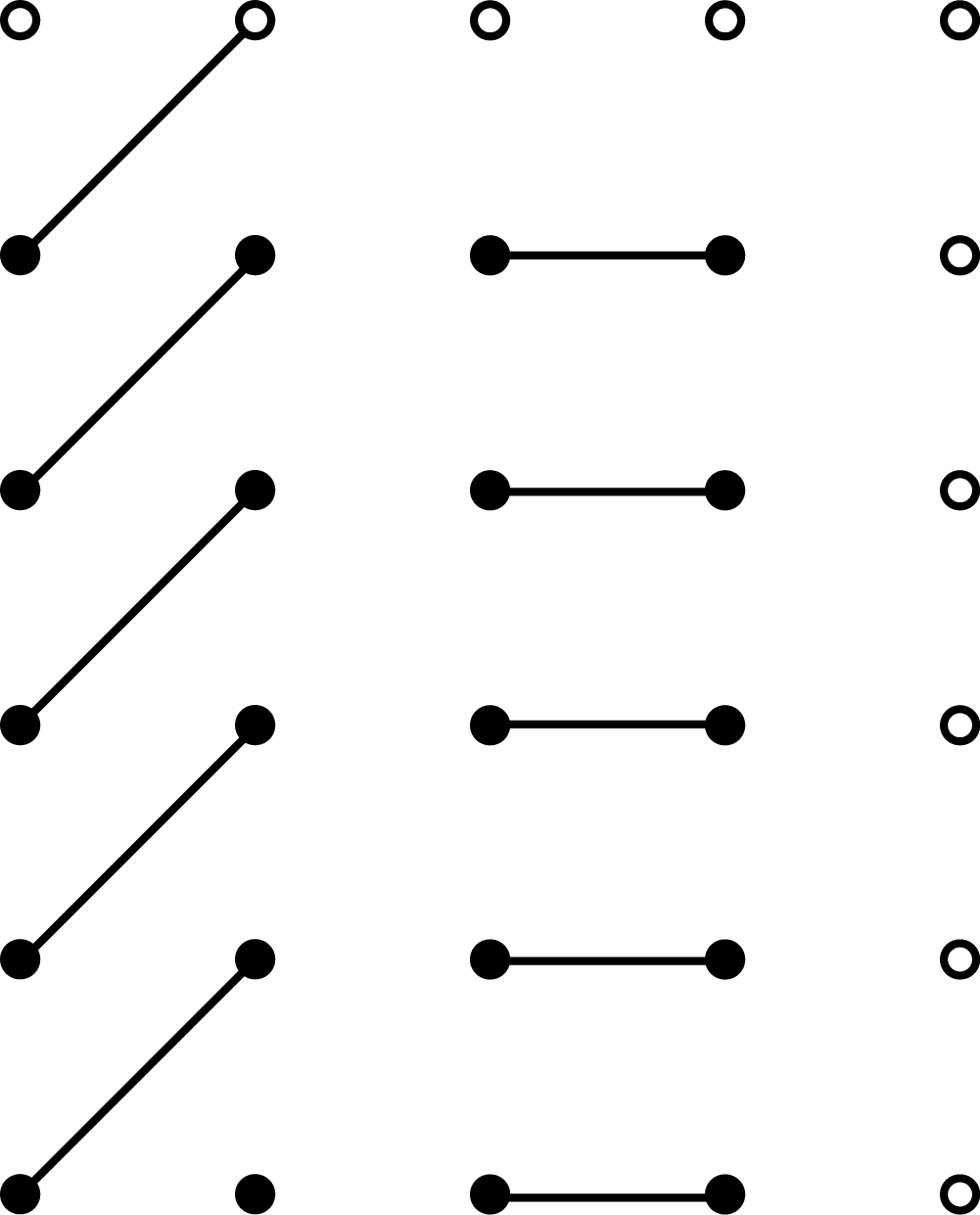}&\hspace{.33in} \includegraphics[scale=.35]{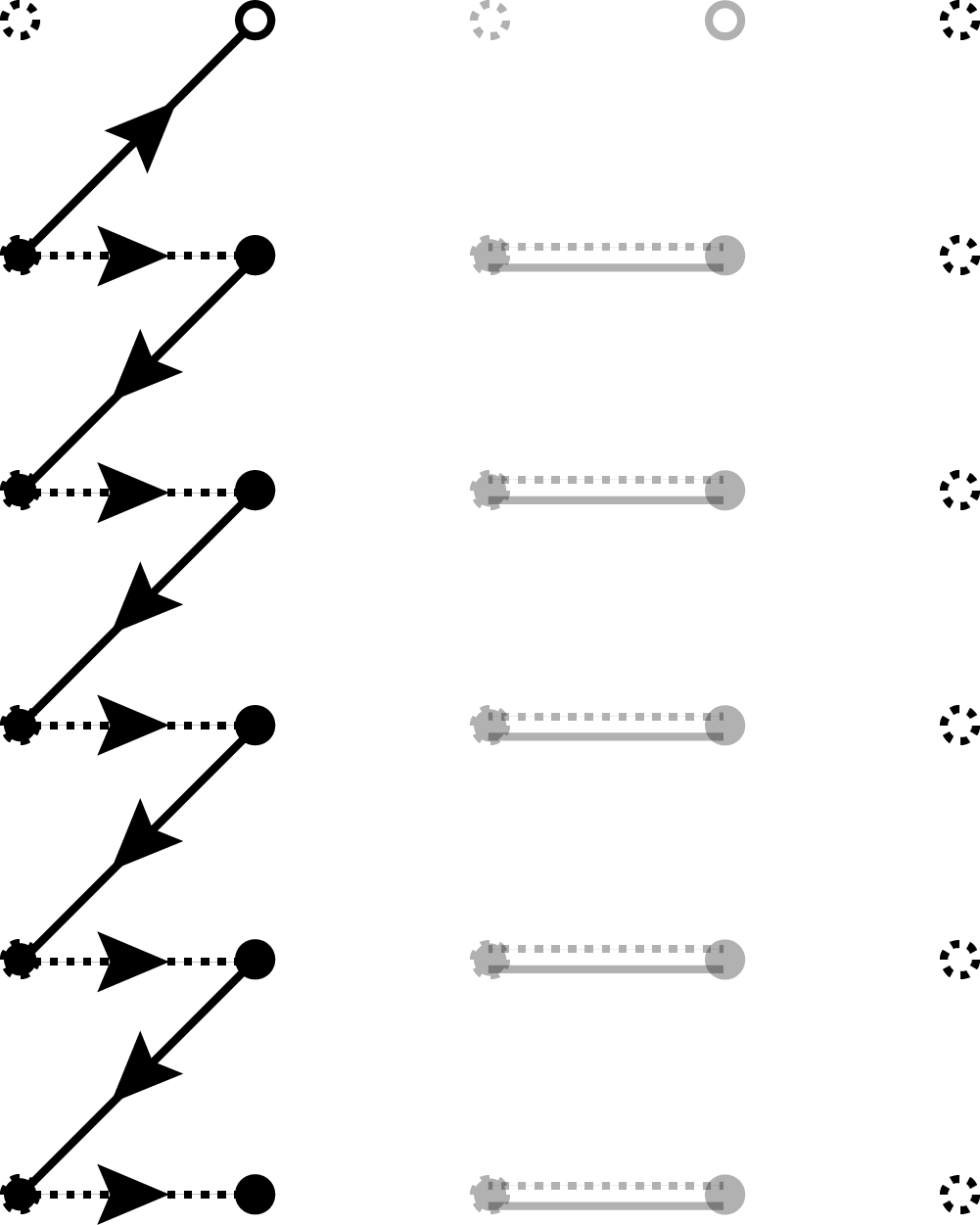}\\  ${}$ \\
(a) $\sg_{\textrm{stack}}\in\Sigma_0$ &\hspace{.2in} (b) $\sg_{\rm stack}\cup\sg_{\textrm{st}}$
\end{tabular}
\caption{The stack configuration and its superposition with the standard configuration  on a $4\times 5$ lattice.}
\label{FPfA2}
\end{figure}

\section{Negativity of $\Pf A_1$} \label{A1}

If $n\equiv 1\pmod 2$, then we have by Theorem \ref{Conj} that $\Pf A_1=-\Pf A_2$. Since  $\Pf A_2>0$, we obtain that $\Pf A_1<0$. Let us turn to the case when $n$ is even. 
We begin with the following case:

\begin{lem} Let $z_h,z_v,z_d>0$. Then if either $n\equiv 2 \pmod 4$ or $m\equiv 2 \pmod 4$, then $\Pf A_1<0$. 
\end{lem}

\begin{proof} First assume $n\equiv 2\pmod 4$ and consider the case $z_h>0$, $z_v=0$, and $z_d>0$. 
By Proposition \ref{detA} (2),
in this case $\det A_1>0$, $\det A_2=0$, hence
\begin{equation}\label{a15a}
Z^{00}-Z^{10}=-Z^{01}-Z^{11} \implies  \Pf A_1= -2Z^{01}-2Z^{11}<0.
\end{equation}
By continuity, $\Pf A_1<0$ for the chosen $z_h>0$, $z_d>0$, and small $z_v>0$.
This proves that $\Pf A_1<0$ in the whole region
$z_h,z_v,z_d>0$.

Now assume $m\equiv 2\pmod 4$ and consider the case $z_h=0$, $z_v>0$, and $z_d>0$. 
By Proposition \ref{detA} (3),
in this case $\det A_1>0$, $\det A_3=0$, hence
\begin{equation}\label{a15b}
Z^{00}-Z^{01}=-Z^{10}-Z^{11} \implies  \Pf A_1= -2Z^{10}-2Z^{11}<0.
\end{equation}
By continuity, $\Pf A_1<0$ for the chosen $z_v>0$, $z_d>0$, and small $z_h>0$.
This proves that $\Pf A_1<0$ in the whole region
$z_h,z_v,z_d>0$.
\end{proof}

The final case $m\equiv 0\pmod 4$, $n\equiv 0\pmod 4$ is the most difficult yet. We consider the asymptotic behavior of $\Pf A_1$ for $z_h=1, 0<z_v\le z_d^2, z_d\to 0$, and we prove the following theorem upon establishing four lemmas:
\begin{theo}\label{theo44.1}
Let $m\equiv 0\pmod 4$, $n\equiv 0\pmod 4$, $z_h=1$, and $0<z_v\le z_d^2$. Then as $z_d\to 0$,\[\Pf A_1=-n^2\left(\frac{m}{2}\right)^nz_v^2z_d^{n-2}\left(1+\mathcal O \left(z_d\right)\right).\]
\end{theo}

\textit{Remark.} This will imply that $\Pf A_1<0$ for $z_h=1, 0<z_v\le z_d^2$, and sufficiently small $z_d$, and hence $\Pf A_1<0$ for all $z_h,z_v,z_d>0$.

First note that for the same reasons as in the beginning discussion of the proof of Lemma \ref{lem1}, configurations with only horizontal dimers on some line $y=k$ cancel each other in the Pfaffian, so we next look at $\Sigma_0$ defined as in \eqref{Sg0}. However, since $n\equiv 0\pmod 4$, configurations in $\Sigma_0$ now cancel each other completely. Indeed, we can decompose $\Sigma_0$ as
\begin{equation}\label{canceldecomp}
\Sigma_0=\Sigma_0^1\sqcup \Sigma_0^2,
\end{equation}
where $\Sigma_0^1$ and $\Sigma_0^2$ can be further decomposed as
\begin{equation}\label{cancel}
\begin{aligned}
\Sigma_0^1&=\bigsqcup_{i=0}^1\Sigma_0^{1,i},\quad \Sigma_0^{1,i}=\{\sigma\in\Sigma_0\, |\,N_{\textrm{diag}}(\sigma;k,k+1)=1,\,j_k\equiv i\pmod 2,\, k\in \mathbb Z_n\},\\
\Sigma_0^2&=\bigsqcup_{i=0}^1 \Sigma_0^{2,i},\quad \Sigma_0^{2,i}=\big\{\sigma\in\Sigma_0\,|\,N_{\textrm{diag}}(\sigma;2k+i,2k+i+1)=2,\, k\in\mathbb Z_{\frac{n}{2}}\big\}.
\end{aligned}
\end{equation}
Here $N_{\textrm{diag}}(\sigma;k,k+1)$ is the number of diagonal dimers in $\sigma$ connecting horizontal line $y=k$ to horizontal line $y=k+1$, and the set of diagonal dimers for every configuration $\sg\in\Sigma_0^1$ is $\{(j_k,k),(j_k+1,k+1)\}$, $k\in\Z_n$. As before, equation \eqref{pos4} holds since all other vertices must be covered by horizontal dimers.
\par On the other hand, configurations in $\Sigma_0^{2,i}$ have the following structure: for every $k\in\Z_{\frac{n}{2}}$ and for $i=0$ or $i=1$ there
is one diagonal dimer which connects a vertex $x_1=(j_1,2k+i)$ to $y_1=(j_1+1,2k+i+1)$ and another diagonal 
dimer which connects a vertex $x_2=(j_2,2k+i)$ to $y_2=(j_2+1,2k+i+1)$. Note that $j_2-j_1\equiv 1\pmod 2$, since all other vertices are
covered by horizontal dimers. We have the following lemma:

\begin{lem}\label{lem44.2}
$|\Sigma_0^1|=|\Sigma_0^2|=2\left(\frac{m}{2}\right)^n$ and for every $\sigma\in\Sigma_0$,\[\sgn(\sigma)=\begin{cases}-1,&\textrm{if }\sigma\in\Sigma_0^1,\\ +1,&\textrm{if }\sigma\in\Sigma_0^2.\end{cases}\]
Hence, $\displaystyle\sum_{\sigma\in\Sigma_0}\sgn(\sg)w(\sg)=0.$
\end{lem}
\begin{proof}
The counting argument for $|\Sigma_0^1|=2\left(\frac{m}{2}\right)^n$ is the same as in the end of the proof of Lemma \ref{lem1}. Similar to the proof of Step 1 in Lemma \ref{lem2}, we have that 

\begin{equation}\label{neg5a1}
\sgn(\sg)=\sgn(\sg_{\rm stack})\;\,\forall\sg\in\Sigma_0^{1,0},
\end{equation}

\begin{equation}\label{neg5b}
\sgn(\sg)=\sgn(\sg_{\rm stack}+\mathsf e_1)\;\,\forall\sg\in\Sigma_0^{1,1}.
\end{equation}
Next we prove that
\begin{equation}\label{neg5c1}
\sgn(\sg_{\rm stack})=\sgn(\sg_{\rm stack}+\mathsf e_1)=-1.
\end{equation}
\noindent Indeed, the superposition $\sg_{\rm stack}\cup\sg_{\rm st}$ consists of trivial contours and one nontrivial contour which is a zigzag path between the vertical lines $x=0$ and $x=1$, similar to Fig.\ \ref{FPfA2} (b) except that the diagonal dimer $\{(0,n-1),(1,0)\}$ has the opposite orientation, hence $\sgn(\sg_{\rm stack})=-1$.

Furthermore, the signs of the superpositions $(\sigma_{\rm stack}+\mathsf e_1)\cup(\sigma_ {\rm st}+\mathsf e_1)$ and $\sigma_{\rm st}\cup(\sigma_{\rm st}+\mathsf e_1)$ are $(-1)$ and $(+1)$, respectively. Indeed, the superposition $(\sigma_{\rm stack}+\mathsf e_1)\cup(\sigma_ {\rm st}+\mathsf e_1)$  consists of trivial contours and a zigzag contour with $n\equiv 0\pmod 2$ arrows in the direction of movement from top to bottom, and so $\sgn((\sigma_{\rm stack}+\mathsf e_1)\cup(\sigma_{\rm st}+\mathsf e_1))=-1.$ In addition, the superposition $\sigma_{\rm st}\cup(\sigma_{\rm st}+\mathsf e_1)$  consists of 
$n$ horizontal contours of the length $m$. Since $m$ is even, the sign of 
each contour is $(-1)$, and
because $n$ is even, the sign of the superposition $\sigma_{\rm st}\cup(\sigma_{\rm st}+\mathsf e_1)$  is $(+1)$.
Since the sign of the superposition 
$(\sigma_{\rm stack}+\mathsf e_1)\cup(\sigma_{\rm st}+\mathsf e_1)$ 
is $(-1)$, and the sign of the superposition $\sigma_{\rm st}\cup(\sigma_{\rm st}+\mathsf e_1)$ is $(+1),$ we have that
$\sgn((\sigma_{\rm stack}+\mathsf e_1)\cup \sigma_{\rm st})=-1.$
Hence, $\sgn(\sigma_{\rm stack}+\mathsf e_1)=-1.$

\par To show that $|\Sigma_0^2|=2\left(\frac{m}{2}\right)^n$, we count as follows: choose either $i=0$ or $i=1$. Now for each $k\in\mathbb Z_{\frac{n}{2}}$, there are $m$ choices for the first diagonal dimer and $\frac{m}{2}$ choices for the second, but the diagonal dimers are equivalent, so we divide this resulting number by 2.
\par Fix $i=0$ and let $\sigma\in\Sigma_0^{2,0}$ be arbitrary. Consider, again, the elementary move described in Step 1 of the proof of Lemma \ref{lem2}, that is, the change of $\sigma$ to $\sigma'$ where a diagonal dimer $\{(j,2k),(j+1,2k+1)\}$ is shifted to $\{(j+2,2k),(j+3,2k+1)\}$ for some $k\in\mathbb Z_{\frac{n}{2}}$. Assume that in $\sigma$ the intermediate vertices $(j+1,2k)$ and $(j+2,2k+1)$ are not occupied by a diagonal dimer. Then the superposition $\sigma\cup\sigma'$ consists of trivial contours and one nontrivial contour of the length 6 comparable to that in Fig. \ \ref{F1} (a), which is positive, hence $\sgn(\sigma)=\sgn(\sigma')$. If the intermediate vertices are occupied by a diagonal dimer, simply apply the clockwise sequence of elementary moves $\frac{m}{2}-1$ times as before. Hence, for every $\sg,\sg'\in\Sigma_0^{2,0}$, $\sgn(\sigma)=\sgn(\sigma')$. Now fix a configuration $\sg_{0,0}\in\Sigma_0^{2,0}$ with diagonal dimer pairs $\{(0,2k),(1,2k+1),(1,2k),(2,2k+1)\}_{k\in\mathbb{Z}_{\frac{n}{2}}}$ and all other dimers horizontal. Then $\sigma_{0,0}\cup\sigma_{\textrm{st}}$ consists of trivial contours and $\frac{n}{2}\equiv 0\pmod 2$ identical nontrivial contours of length $m+2$ as shown in Fig.\ \ref{F2}, hence $\sgn(\sigma)=\sgn(\sigma_{\textrm{st}})$. This shows that for any $\sigma\in\Sigma_0^{2,0}$, $\sgn(\sigma)=+1$. The proof in the case of $i=1$ is entirely analogous. Since $\Sigma_0^1$ and $\Sigma_0^2$ have the same cardinality but have configurations with opposite signs, they cancel each other and the claim follows.
\begin{figure}
\includegraphics[scale=.5]{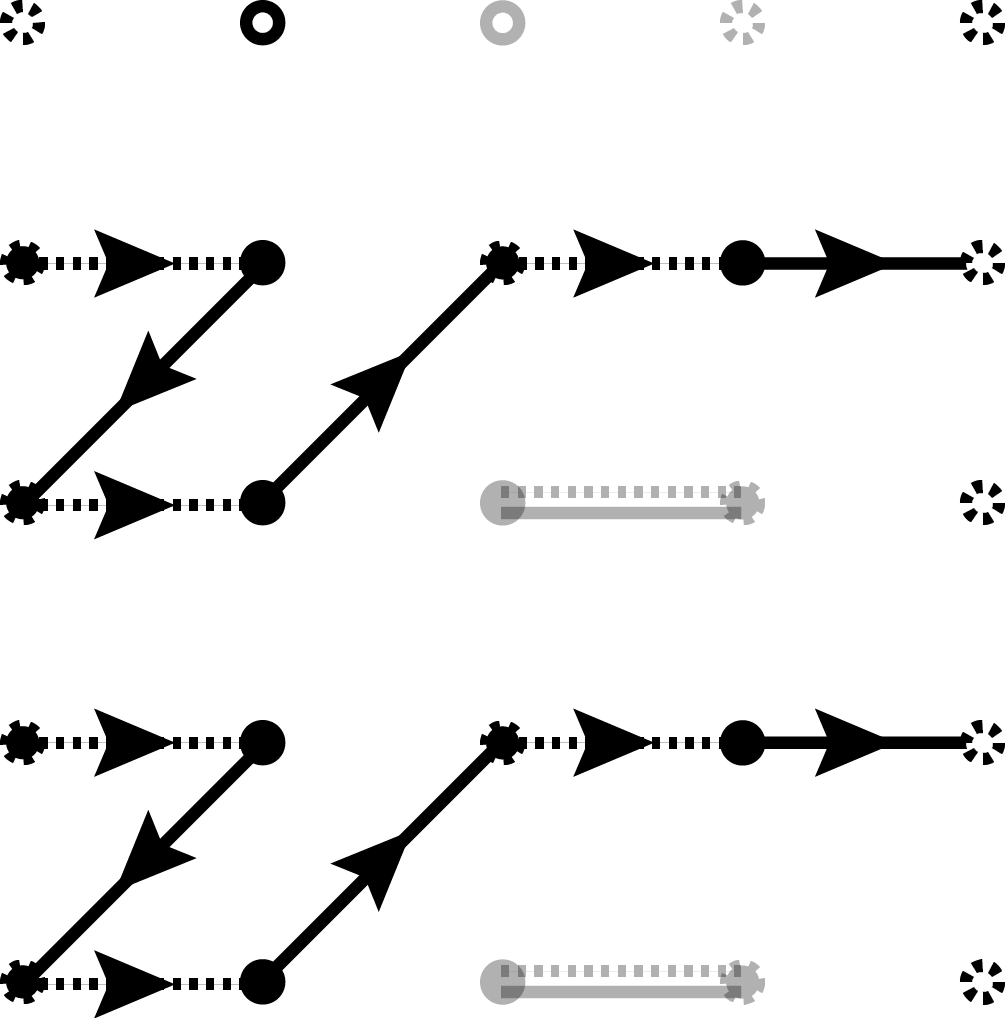}
\caption{$\sg_{0,0}\cup\sg_{\textrm{st}}$}
\label{F2}
\end{figure}
\end{proof}
Let us consider then the set of configurations 
\[\Sigma_1=\{\sg\in\Sigma_{m,n}\,|\, N_{\textrm{d,v}}(\sg,k)=2, \, k\in\Z_n;\,N_{\textrm{diag}}(\sg)=n-1, N_{\textrm{vert}}(\sg)=1\},\]
where $N_{\textrm{d,v}}(\sg,k)$ is the total number of vertices covered by diagonal and vertical dimers in $\sg$ on the line $y=k$, and $N_{\textrm{vert}}(\sg)$ is the number of vertical dimers in $\sg$. However, by the following lemma this set is empty:
\begin{lem}\label{s1empty} 
$\Sigma_1=\emptyset.$
\end{lem}
\begin{proof}
Assume without loss of generality that the vertical dimer is $\{(0,n-1),(0,0)\}$ and let $\{(j_k,k),(j_k+1,k+1)\}_{k\in\mathbb Z_{n-1}}$ denote the set of diagonal dimers. Since we require all other vertices to be covered by horizontal dimers, we have that 
\begin{enumerate}\label{ee}
\item{\makebox[8cm]{$j_k+1-j_{k+1}\equiv 1\pmod 2, \quad\forall k\in\mathbb Z_{n-2}\hfill\Rightarrow$} $j_{k+1}-j_{k}\equiv 0\pmod 2,\quad\forall k\in\mathbb Z_{n-2}$},
\item{\makebox[8cm]{$j_{0}-0\equiv 1\pmod 2\hfill\Rightarrow$} $j_0\equiv 1\pmod 2$},
\item{\makebox[8cm]{$j_{n-2}+1-0\equiv 1\pmod 2\hfill\Rightarrow$} $j_{n-2}\equiv 0\pmod 2$},
\end{enumerate}
where (1) is the requirement that diagonal vertices within a given horizontal line be odd spacing apart, (2) is the requirement that $(0,0)$, the top vertex of the vertical dimer, be odd spacing from $(j_0,0)$, the bottom vertex of the diagonal dimer $\{(j_0,0),(j_0+1,1)\}$, and (3) is the requirement that $(0,n-1)$, the bottom vertex of the vertical dimer, be odd spacing apart from $(j_{n-2}+1,n-1)$, the top vertex of the diagonal dimer $\{(j_{n-2},n-2),(j_{n-2}+1,n-1)\}$. From (1) it follows that \[j_{n-2}-j_{n-3}+j_{n-3}-j_{n-4}+j_{n-4}-\ldots+j_3-j_2+j_2-j_1+j_1-j_0\equiv 0\pmod 2.\] Hence, $j_{n-2}\equiv j_0 \pmod 2$. Combining this and (2) gives us that $j_{n-2}\equiv 1\pmod 2$, which contradicts (3).  
\end{proof}

Since $\Sigma_1$ is empty, let us consider
\[\Sigma_2=\{\sg\in\Sigma_{m,n}\,|\, N_{\textrm{d,v}}(\sg,k)=2,
 \, k\in\Z_n; N_{\textrm{diag}}(\sg)=n-2, N_{\textrm{vert}}(\sg)=2\}.\]
Analogous to the decomposition of $\Sigma_0$ above, let us write 
\begin{equation}\label{s2decomp}
\Sigma_2=\Sigma_2^1\sqcup \Sigma_2^2,
\end{equation}
where 
\begin{equation}
\begin{aligned}
\Sigma_2^1&=\{\sigma\in\Sigma_2\,|\,N_{\textrm{d,v}}(\sg;k,k+1)=1, \,k\in \mathbb Z_n\},\\
\Sigma_2^2&=\bigsqcup_{i=0}^1\Sigma_2^{2,i},\quad \Sigma_2^{2,i}=\{\sigma\in\Sigma_2\,|\,N_{\textrm{d,v}}(\sg;2k+i,2k+i+1)=2,\, k\in \mathbb Z_{\frac{n}{2}}\},
\end{aligned}
\end{equation}
i.e. $\Sigma_2^1$ and $\Sigma_2^2$ consist respectively of configurations whose vertical and diagonal dimers connect all horizontal lines to one another and of configurations whose vertical and diagonal dimers are pairwise placed on horizontal lines $y=2k+i$ and $y=2k+1+i$, $k\in\mathbb Z_{\frac{n}{2}}$ for $i=0$ or $1$. We again require that every pair of vertices of diagonal and/or vertical dimers in a given horizontal line be an odd spacing apart since the remaining dimers must be horizontal. See Fig.\ \ref{F7} (a), (b), and (c) for examples of configurations $\sg\in\Sigma_2^1$, $\sg\in\Sigma_2^{2,0}$, and $\sg\in\Sigma_2^{2,1}$, respectively.

\begin{figure}[h]
\begin{tabular}{c c c}
\includegraphics[scale=.45]{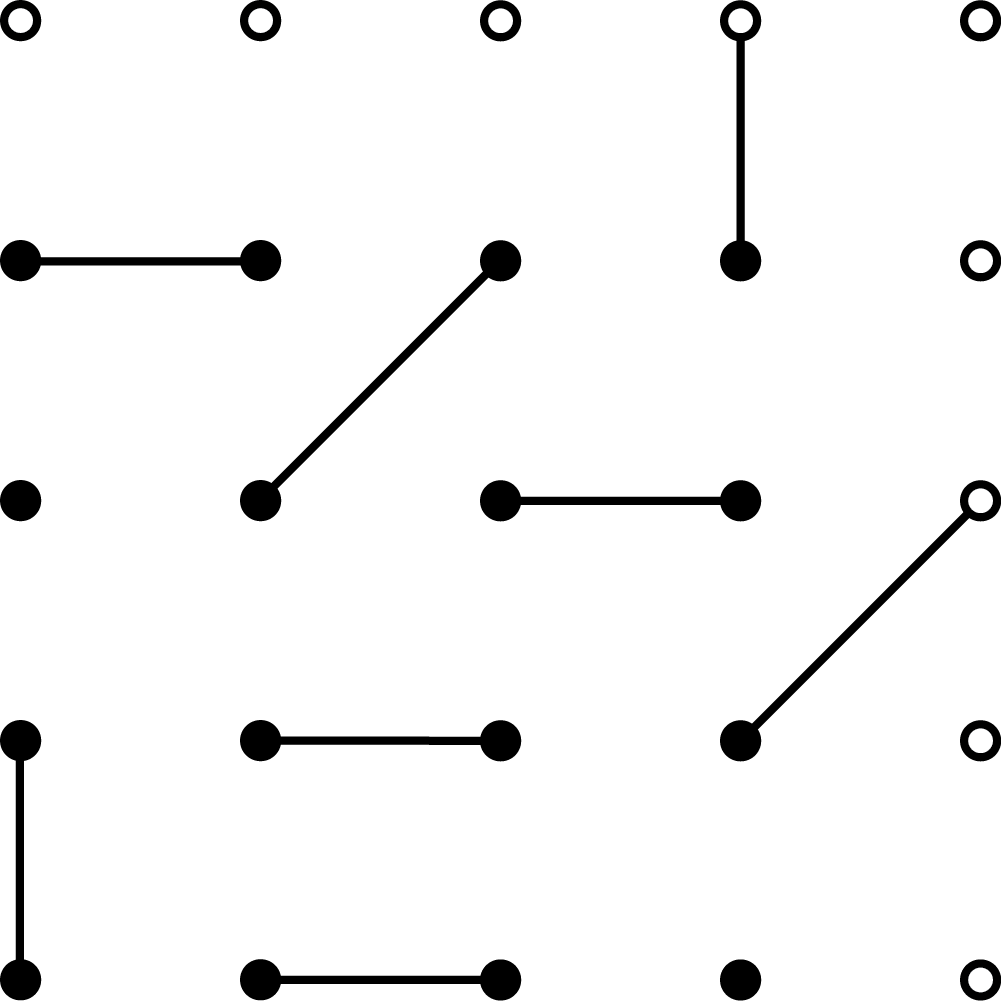}&\hspace{.45in} \includegraphics[scale=.45]{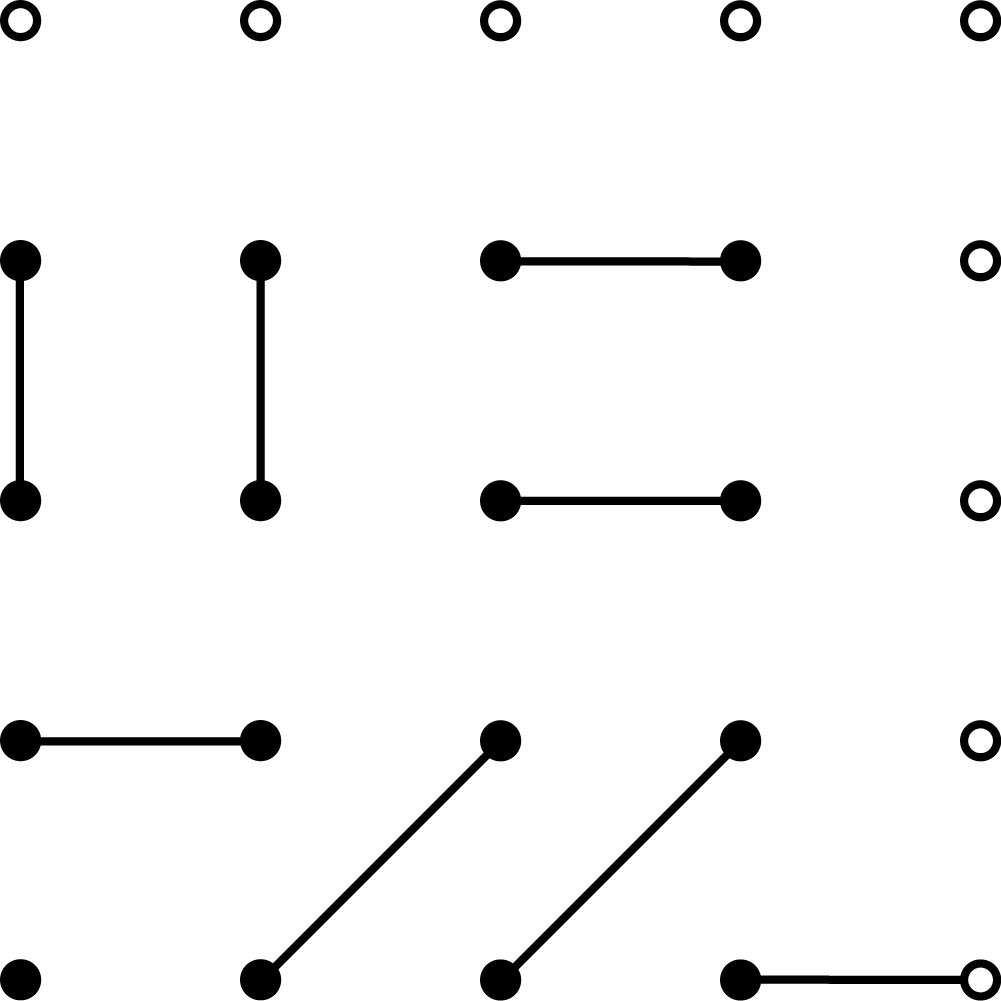}&\hspace{.45in}
\includegraphics[scale=.45]{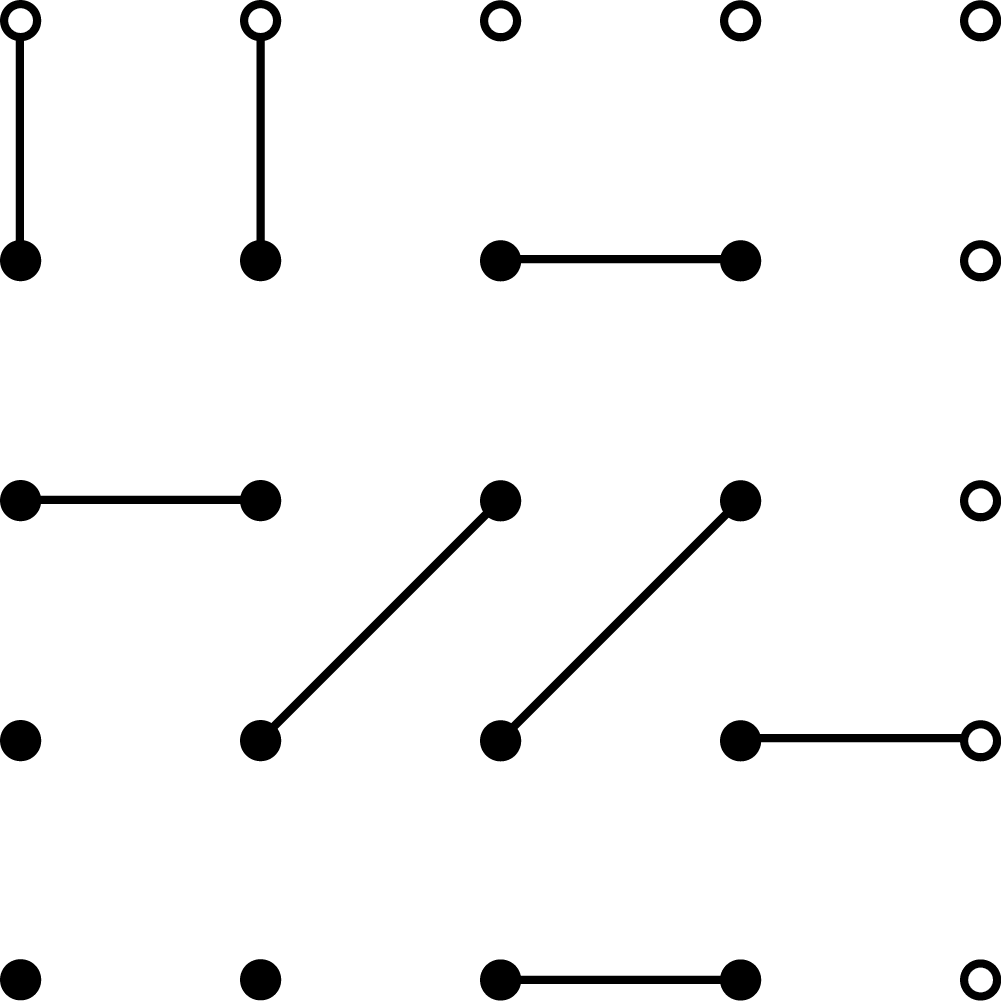}\\
(a) $\sg\in\Sigma_2^1$ &\hspace{.2in} (b) $\sg\in\Sigma_2^{2,0}$ &\hspace{.2in} (c) $\sg\in\Sigma_2^{2,1}$
\end{tabular}
\caption{Examples of configurations from $\Sigma_2$.}
\label{F7}
\end{figure}

We have the following lemma:
\begin{lem}\label{lem44.3}
$|\Sigma_2^1|=n(n-1)\left(\frac{m}{2}\right)^n$, and $\sgn(\sigma)=-1 \text{ for every } \sg\in\Sigma_2^1$.
\end{lem}
\begin{proof}
There are $\binom{n}{2}$ choices of horizontal lines for the vertical dimers, $m$ choices to place the first vertical dimer within its horizontal line and $\frac{m}{2}$ choices to place each of the second vertical dimer and all $n-2$ diagonal dimers, hence $|\Sigma_2^1|=n(n-1)\left(\frac{m}{2}\right)^n$.

Let us prove that $\sgn(\sigma)=-1$  for every $\sg\in\Sigma_2^1$. We have seen in the proof of Lemma \ref{lem2} that the sign of a configuration is invariant with respect to 
{\it elementary moves}, horizontal shifts of diagonal dimers by two units to the right.
The same reasoning applies to vertical dimers as well.

In $\sg\in\Sigma_2^1$ define an \textit{elementary swap} that interchanges a diagonal dimer $\{(j,k),(j+1,k+1)\}$ with a vertical dimer $\{(j+1,k-1),(j+1,k)\}$ below it as follows: increase the $j$-coordinate of every vertex of every dimer along the horizontal line $y=k$, $k\in \Z_n,$ to produce a new configuration $\sigma'$ in which the vertical dimer is now above the diagonal dimer. The superposition $\sigma\cup\sigma'$ consists of trivial contours and one nontrivial contour of length $m+2$, $m-1\equiv 1\pmod 2$ of whose arrows are in the direction of movement from left to right if $j\equiv 1\pmod 2$ and $m+1\equiv 1\pmod 2$ of whose arrows are in the direction of movement from left to right if $j\equiv 0\pmod 2$, the former case exemplified in Fig.\ \ref{F3} with $j=1,\,k=2,\,m=4$. Hence, for every $\sg,\sg'\in\Sigma_2^1,\, \sgn(\sg)=\sgn(\sg')$. 
In other words, for any elementary move (horizontal shifts to the right by two or swaps) $\sg\to\sg'$ we have that formula \eqref{sgn-elem} holds. 

We now show that any configuration $\sg\in\Sigma_2^1$ can be moved to a configuration $\sg_{\rm stack}$ in which vertical dimers are placed at positions $\{(0,0),(0,1)\}$ and $\{(1,1),(1,2)\},$ respectively, whereas diagonal dimers form a stack above these two vertical dimers and between lines $x=0$ and $x=1.$ The remaining vertices are occupied by horizontal dimers. See Fig.\ \ref{F3a} (a).

Fix $\sg\in\Sigma_2^1$ and assume that vertical dimers are at positions $\{(j_1,k_1),(j_1,k_1+1)\}$ and $\{(j_2,k_2),(j_2,k_2+1)\},$ respectively. Note that 
\begin{equation}\label{formula}
j_2-j_1\equiv 1 \pmod 2.
\end{equation}
For the sake of contradiction, let us assume that $j_2-j_1\equiv 0 \pmod 2.$ If these two vertical dimers are not on the same horizontal line $y=k,$ apply sequence of swaps on one of them until $k_1=k_2.$ Note that in order to swap a diagonal dimer with a vertical dimer, $\{(j,k),(j,k+1)\},$ one vertex of a diagonal dimer has to be on the vertical line $x=j$ and the other vertex on a horizontal line $y=k+1.$ If this were not the case, horizontally shift diagonal dimer until this property is satisfied. Therefore, let us assume that after swaps vertical dimers are positioned at $\{(j_1,k),(j_1,k+1)\}$ and $\{(j_2,k),(j_2,k+1)\},$ respectively. Now, $j_2-j_1\equiv 0 \pmod 2$ implies that along horizontal line $y=k$ there is an odd number of vertices between $(j_1,k)$ and $(j_2,k)$, but since all other vertices on $y=k$ have to be covered by horizontal dimers we have a contradiction. Therefore, equation \eqref{formula} holds. Further, note that this means that two vertical dimers  cannot be placed along the same vertical line $x=j.$ 

In other words, we may assume that $j_1\equiv 0 \pmod 2$ and $j_2\equiv 1 \pmod 2.$ If for a vertical dimer $\{(j_1,k_1),(j_1,k_1+1)\}$, $k_1\neq 0,$ then use swaps and horizontal shifts of diagonal dimers where necessary until this dimer is at line $y=0.$ In a similar fashion, apply swaps with necessary shifts of diagonal dimers until vertical dimer $\{(j_2,k_2),(j_2,k_2+1)\}$ is at line $y=1.$ Now, shift these two vertical dimers horizontally until they are positioned at $\{(0,0),(0,1)\}$ and $\{(1,1),(1,2)\}$, respectively. If it so happens that after the above elementary moves they are positioned at $\{(0,1),(0,2)\}$ and $\{(1,0),(1,1)\}$, respectively, apply swaps until they are in the desired position. Similarly to the proof of Lemma \ref{lem2}, by horizontal shifts by two to the right, we can first move the diagonal dimer at horizontal line $y=2$ to position  $\{(0,2), (1,3)\}$ and then inductively each diagonal dimer on each of the lines $y=k,$ $k = 3, 4, \ldots, n-1,$ to position  $\{(0,k), (1,k+1)\}$, forming a stack of diagonal dimers above the vertical dimers. Hence, we have obtained a configuration $\sg_{\rm stack}.$

In the superposition $\sg_{\rm stack}\cup\sg_{\textrm{st}}$, there are trivial contours and one nontrivial contour of length $2n$, $2n-2\equiv 0 \pmod 2$ of whose arrows are in the direction of movement from top to bottom as shown in Fig.\ \ref{F3} (b). Hence, $\sgn(\sg_{\rm stack})=-1$ and the claim follows.
\begin{figure}[h]
\begin{tabular}{c c}
\includegraphics[scale=.5]{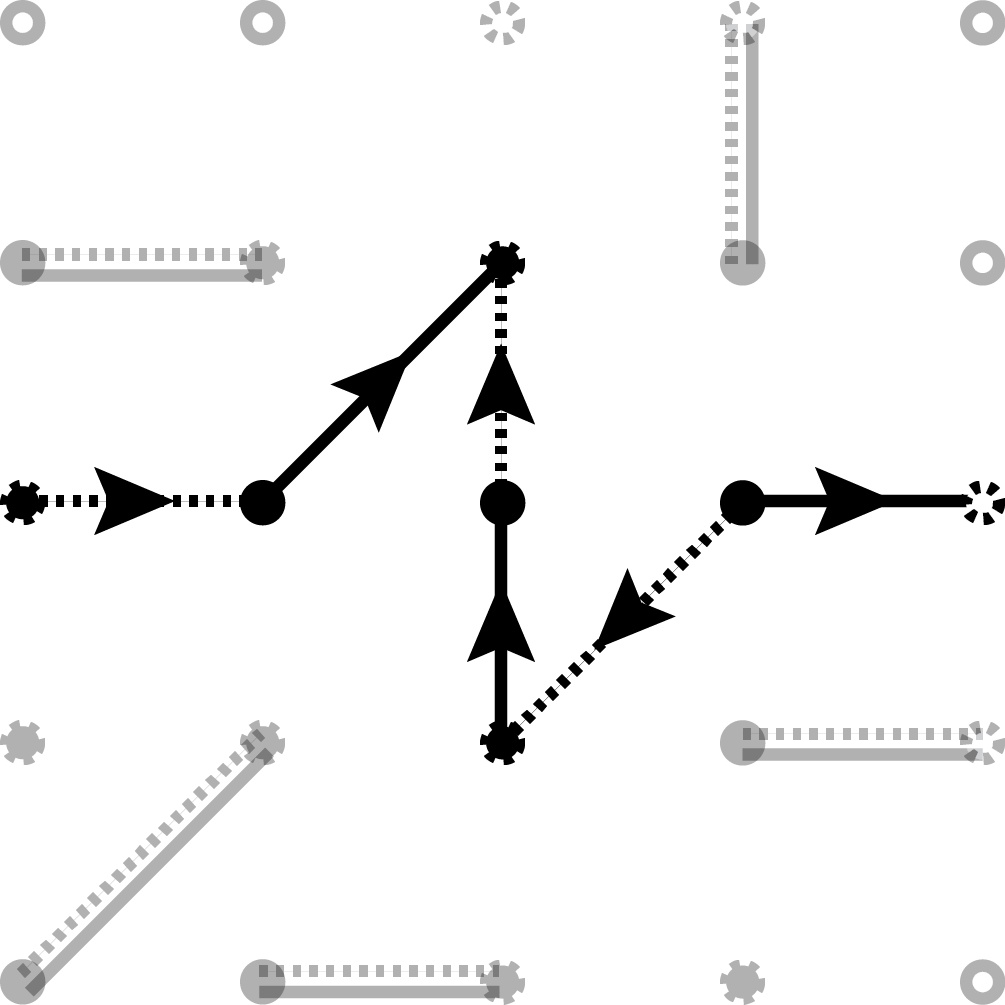} 
\\
$\sg\cup\sg'$ &\hspace{.5in} 
\end{tabular}
\caption{Example of a type of superposition on a $4\times 4$ lattice.}
\label{F3}
\end{figure}

\begin{figure}[h]
\begin{tabular}{c c}
\includegraphics[scale=.5]{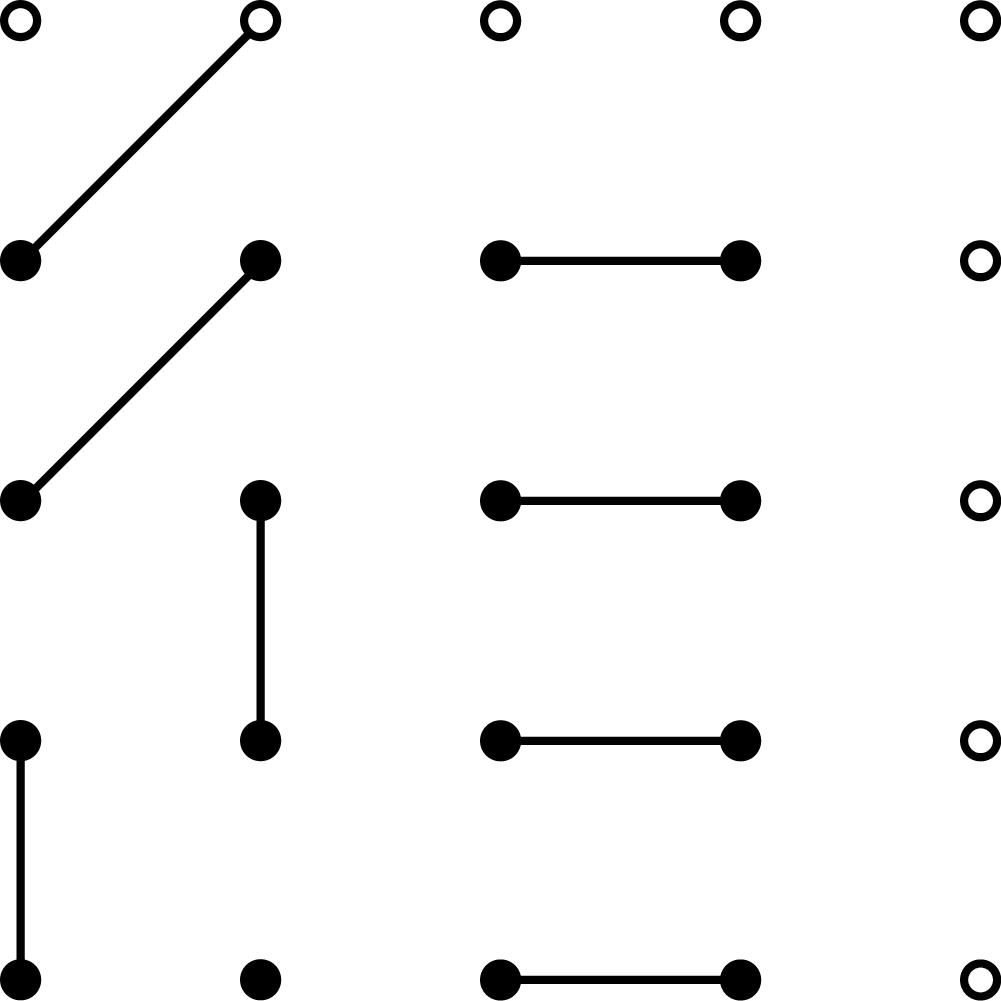}&\hspace{.5in} \includegraphics[scale=.5]{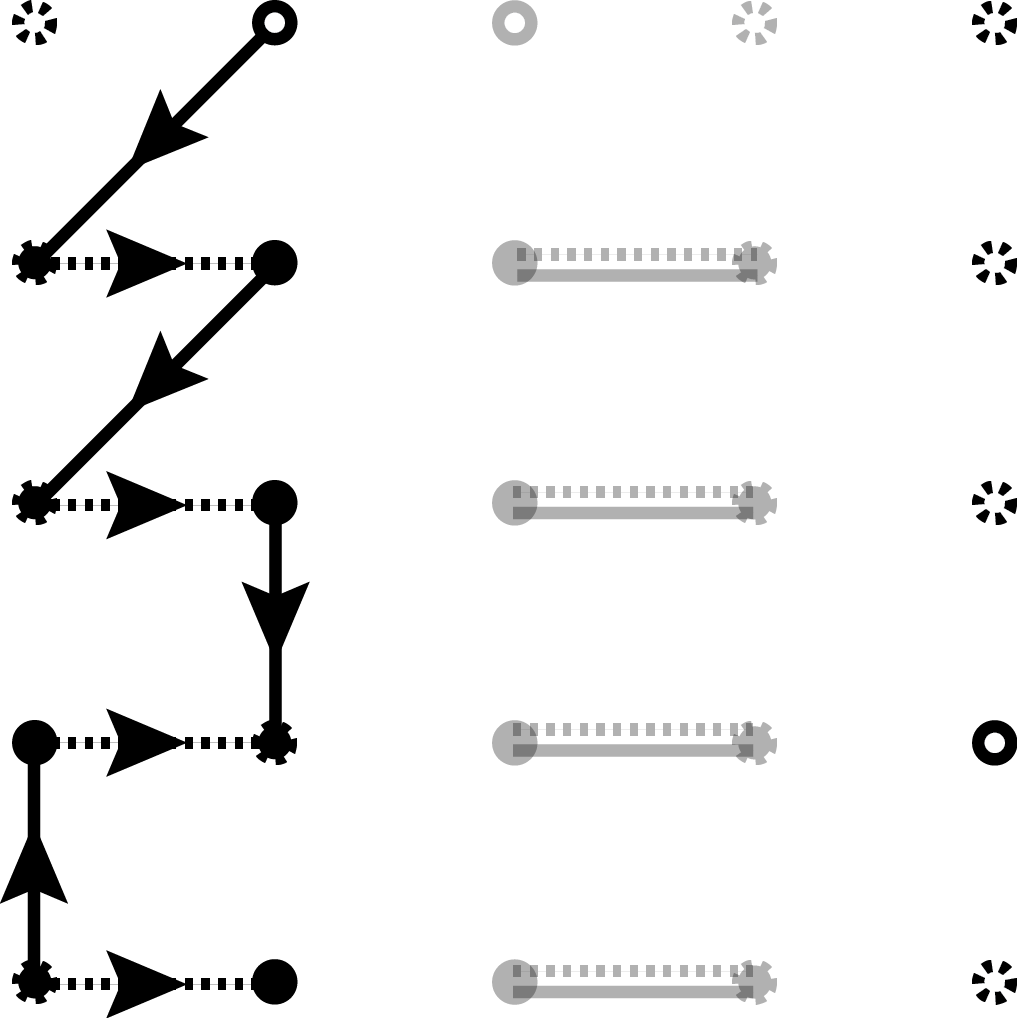}\\
(a) $\sg_{\rm stack}\in\Sigma_2^1$ &\hspace{.5in} (b) $\sg_{\rm stack}\cup\sg_{\textrm{st}}$ 
\end{tabular}
\caption{The stack configuration $\sg_{\rm stack}\in\Sigma_2^1$ and its superposition with the standard configuration $\sg_{\textrm{st}}$.}
\label{F3a}
\end{figure}
\end{proof}

\begin{lem}\label{lem44.4}
$|\Sigma_2^2|=n\left(\frac{m}{2}\right)^n$, and $\sgn(\sigma)=-1\text{ for every }\sg\in\Sigma_2^2$.
\end{lem}
\begin{proof}
There are $n$ choices of horizontal lines within which to place the lower vertices of the two vertical dimers, and within this horizontal line there are $m$ choices to place the first vertical dimer and $\frac{m}{2}$ choices to place the second, but the dimers are equivalent so we divide this number by $2$. Subsequently, there are likewise $\left(\frac{m}{2}\right)^2$ choices in placing each of the $\frac{n}{2}-1$ diagonal dimer pairs, hence $|\Sigma_2^2|=n\left(\frac{m}{2}\right)^n$.

For any given configuration $\sg\in\Sigma_2^2,$ by applying  a sequence of 
{\it elementary moves}, i.e. horizontal shifts of diagonal and vertical dimers by two units to the right, we can position all vertical and diagonal dimers between the vertical lines $x=0$ and $x=1$ to obtain a configuration $\sg_{\rm stack}.$ See Fig.\ 
\ref{F4} (a). Since the elementary moves do not change the sign of configuration, 
we obtain that $\sgn(\sg)=\sgn(\sg_{\rm stack}).$ 

The superposition $\sg_{\rm stack}\cup\sg_{\textrm{st}}$ contains trivial contours, one nontrivial positive square-shaped contour, and $\frac{n}{2}-1\equiv 1\pmod 2$ negative contours as shown in Fig.\ \ref{F4} (b). Hence, $\sgn(\sigma_{\rm stack})=-1.$

\begin{figure}[h]
\begin{tabular}{c c}
\includegraphics[scale=.5]{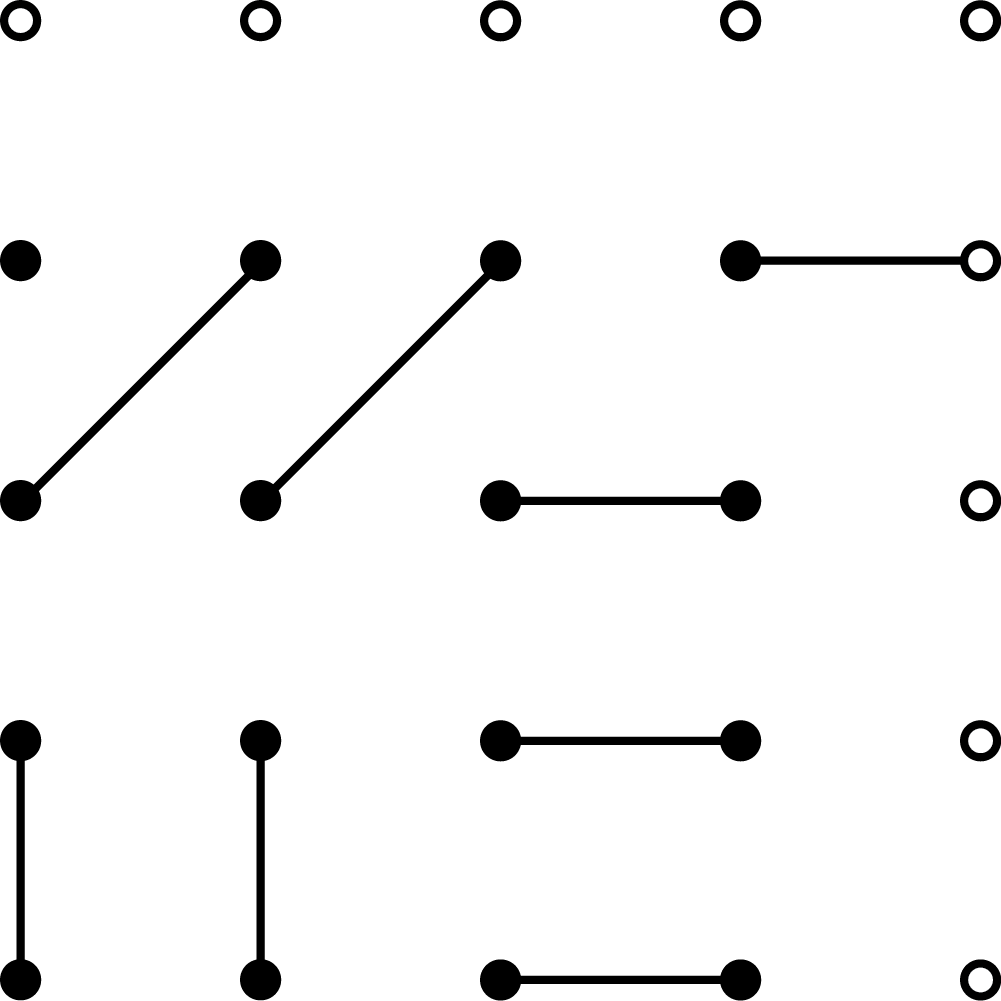}&\hspace{.5in} 
\includegraphics[scale=.5]{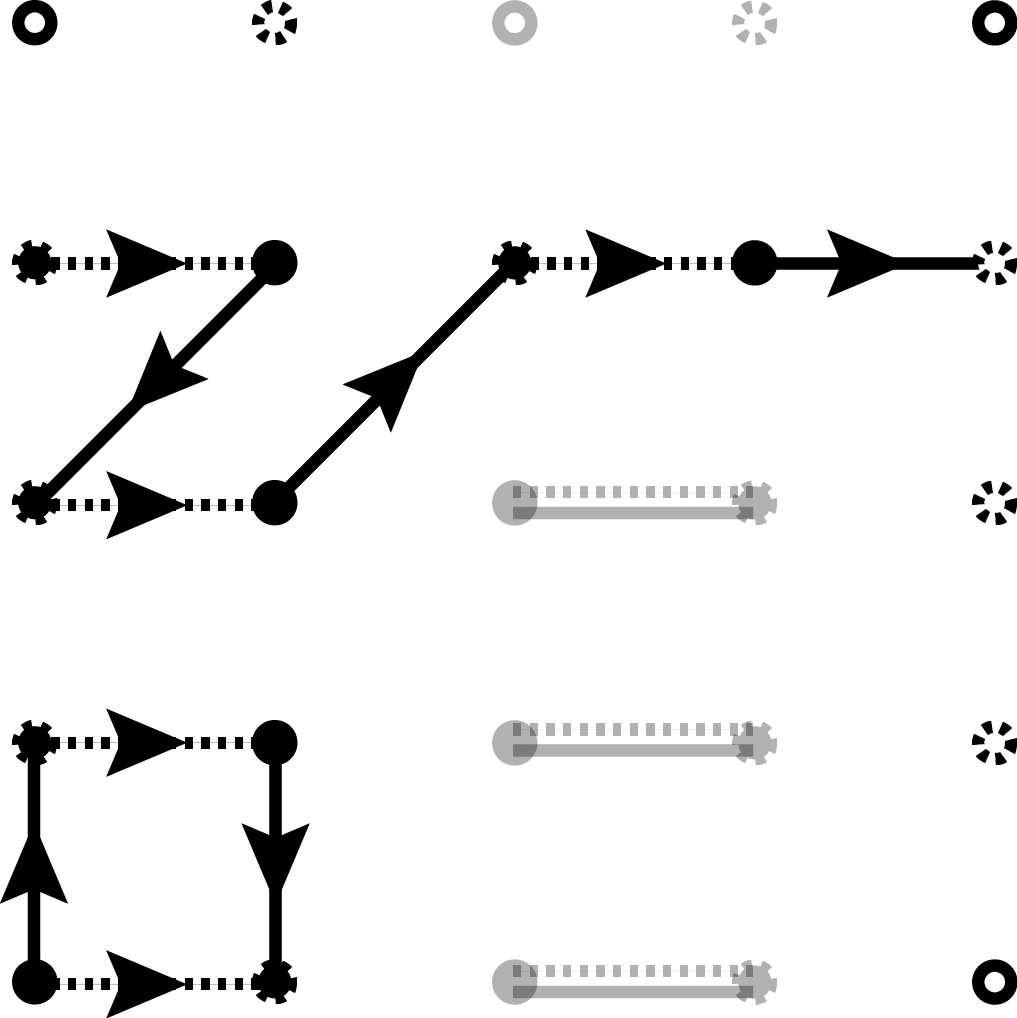}\\
(a) $\sg\in \Sigma_2^2$ &\hspace{.5in} 
(b) $\sg_{\rm stack}\cup\sg_{\textrm{st}}$ 
\end{tabular}
\caption{Examples of types of superpositions on a $4\times 4$ lattice.}
\label{F4}
\end{figure}
\end{proof}

{\it Proof of Theorem \ref{theo44.1}}
We have that
\begin{equation}\label{PfA1}
\Pf A_1=\sum_{\sg\in\Sigma_{m,n}}\sgn(\sg)w(\sg).
\end{equation}
We assume here that $z_h=1$, $z_v\le z_d^2$, and $z_d$ is small. For the same reasons as those given in Lemma \ref{lem1}, we have that 
\begin{equation}
\sum_{\sg\in\cup_{k\in\mathbb Z_n}\Sigma(k)}\sgn(\sg)w(\sg)=0,
\end{equation}
hence to get lowest excited states without cancellations we next consider $\Sigma_0$. However, by Lemma \ref{lem44.2} 
\begin{equation}
\sum_{\sg\in\Sigma_0}\sgn(\sg)w(\sg)=0,
\end{equation}
so we next turn to $\Sigma_1$, but this is empty by Lemma \ref{s1empty}. Therefore, we finally consider $\Sigma_2$, which, by Lemmas \ref{lem44.3} and \ref{lem44.4} yields
\begin{equation}
\begin{aligned}
\Pf A_1=\sum_{\sg\in\Sigma_{m,n}}\sgn(\sg)w(\sg)&=\sum_{\sg\in \Sigma_2}\sgn(\sg)w(\sg)+\sum_{\sg\in\Sigma_{m,n}\setminus\hat{\Sigma}}\sgn(\sg)w(\sg)\\
&=-n^2\left(\frac{m}{2}\right)^nz_v^2z_d^{n-2}+\sum_{\sg\in\Sigma_{m,n}\setminus\hat{\Sigma}}\sgn(\sg)w(\sg),
\end{aligned}
\end{equation}
where $\hat{\Sigma}=\cup_{i=0}^2\Sigma_i\cup\cup_{k\in\mathbb Z_n}\Sigma(k)$. However, any configuration in $\Sigma_{m,n}\setminus\hat{\Sigma}$ will either have a greater total number of vertical and diagonal dimers or it will have a weight 
\begin{equation}
\frac{z_v^\ell}{z_d^\ell}z_v^2z_d^{n-2}\le z_d^\ell z_v^2z_d^{n-2},\quad 2\le \ell \le n-2,
\end{equation}
from which the claim follows.$\hfill \square$

\section{Poisson Summation Formula and Asymptotics of The Partition Function}\label{Poisson}

As an application of the Pfaffian Sign Theorem, we prove the following theorem:

\begin{theo} Suppose that $ m,n \to \infty$ in such a way that
\begin{equation}\label{c12}
C_1\le \frac{m}{n}\le C_2
\end{equation}
for some positive constants $C_2>C_1$. Then for some $c>0$,
\begin{equation}\label{AsympF}
Z=2e^{\frac{1}{2}mnF}\left(1+\mathcal{O}\left(e^{-c(m+n)}\right)\right),
\end{equation}
where 
\begin{equation}\label{p16}
F=\ln 2+\int\limits_0^1\int\limits_0^1 f(x,y)\,dx\,dy,
\end{equation}
and
\begin{equation}\label{pcopy2}
f(x,y)=\frac{1}{2}\,\ln\big[z_h^2 \sin^2 (2\pi x)
+z_v^2 \sin^2 (2\pi y)
+z_d^2\cos^2 (2\pi x+2\pi y)\big].
\end{equation}
\end{theo}

\begin{proof}
\noindent By the double product formula,
\begin{equation}\label{p1}
\frac{1}{mn}\ln\det A_i=\ln 2+\frac{1}{mn}\sum_{j=0}^{m-1}\sum_{k=0}^{n-1}
f(x_j,y_k),\quad x_j=\frac{j+\al_{i}}{m}\,,\quad y_k=\frac{k+\be_{i}}{n}\,,
\end{equation}
with
\begin{equation}\label{p2}
f(x,y)=\frac{1}{2}\,\ln\big[z_h^2 \sin^2 (2\pi x)
+z_v^2 \sin^2 (2\pi y)
+z_d^2\cos^2 (2\pi x+2\pi y)\big].
\end{equation}
Note that in \eqref{p1} we were able to change the upper limit in the first sum from $\frac{m}{2}-1$ to $m-1$ due to the symmetries of $\sin^2 x$ and $\cos^2 x$ functions.
The sum in \eqref{p1} is a Riemann sum and we evaluate its asymptotics by the Poisson summation formula.
To that end, since $f(x+\frac{1}{2},y)=f(x,y)$ and $f(x,y+\frac{1}{2})=f(x,y)$, we expand $f(x,y)$ into Fourier series 
\begin{equation}\label{p3}
f(x,y)=\sum_{(p,q)\in\Z^2} a(p,q) e^{2\pi i (px+qy)},
\end{equation}
where
\begin{equation}\label{p4}
a(p,q)=\int\limits_0^1\int\limits_0^1 f(x,y)e^{-2\pi i (px+qy)}\,dx\,dy.
\end{equation}
Then
\begin{equation}\label{p5}
\frac{1}{mn}\sum_{j=0}^{m-1}\sum_{k=0}^{n-1}
f(x_j,y_k)=\sum_{(p,q)\in\Z^2}\frac{a(p,q)}{mn}\sum_{j=0}^{m-1}\sum_{k=0}^{n-1}
  e^{2\pi i (px_j+qy_k)}.
\end{equation}
Note that
\begin{equation}\label{p6}
\frac{1}{m}\sum_{j=0}^{m-1}e^{2\pi i px_j}=e^{2\pi i p\frac{\al_{i}}{m}}\frac{1}{m}\sum_{j=0}^{m-1}e^{2\pi i p\frac{j}{m}}=e^{2\pi i p\frac{\al_{i}}{m}}\times\begin{cases}0,&p\not\equiv 0\pmod m\\ 1,&p\equiv 0\pmod m\end{cases},
\end{equation}
and similarly, 
\begin{equation}\label{p7}
\frac{1}{n}\sum_{k=0}^{n-1}e^{2\pi i qy_k}=e^{2\pi i q\frac{\be_{i}}{n}}\frac{1}{n}\sum_{k=0}^{n-1}e^{2\pi i q\frac{k}{n}}=e^{2\pi i q\frac{\be_{i}}{n}}\times\begin{cases}0,&q\not\equiv 0\pmod n\\ 1,&q\equiv 0\pmod n\end{cases}.
\end{equation}
By substituting equations \eqref{p6} and \eqref{p7} into equation \eqref{p5} we obtain the Poisson summation formula
\begin{equation}\label{p8}
\begin{aligned}
\frac{1}{mn}\sum_{j=0}^{m-1}\sum_{k=0}^{n-1}
f(x_j,y_k)&=\sum_{(s,t)\in\Z^2}a(ms,nt)e^{2\pi i (s\al_{i}+t\be_{i})}\\
&=a(0,0)+\sideset{}{'}\sum_{(s,t)\in\Z^2}a(ms,nt)e^{2\pi i (s\al_{i}+t\be_{i})}.
\end{aligned}
\end{equation}

\begin{figure}[h]
\centering
\includegraphics[scale=.5]{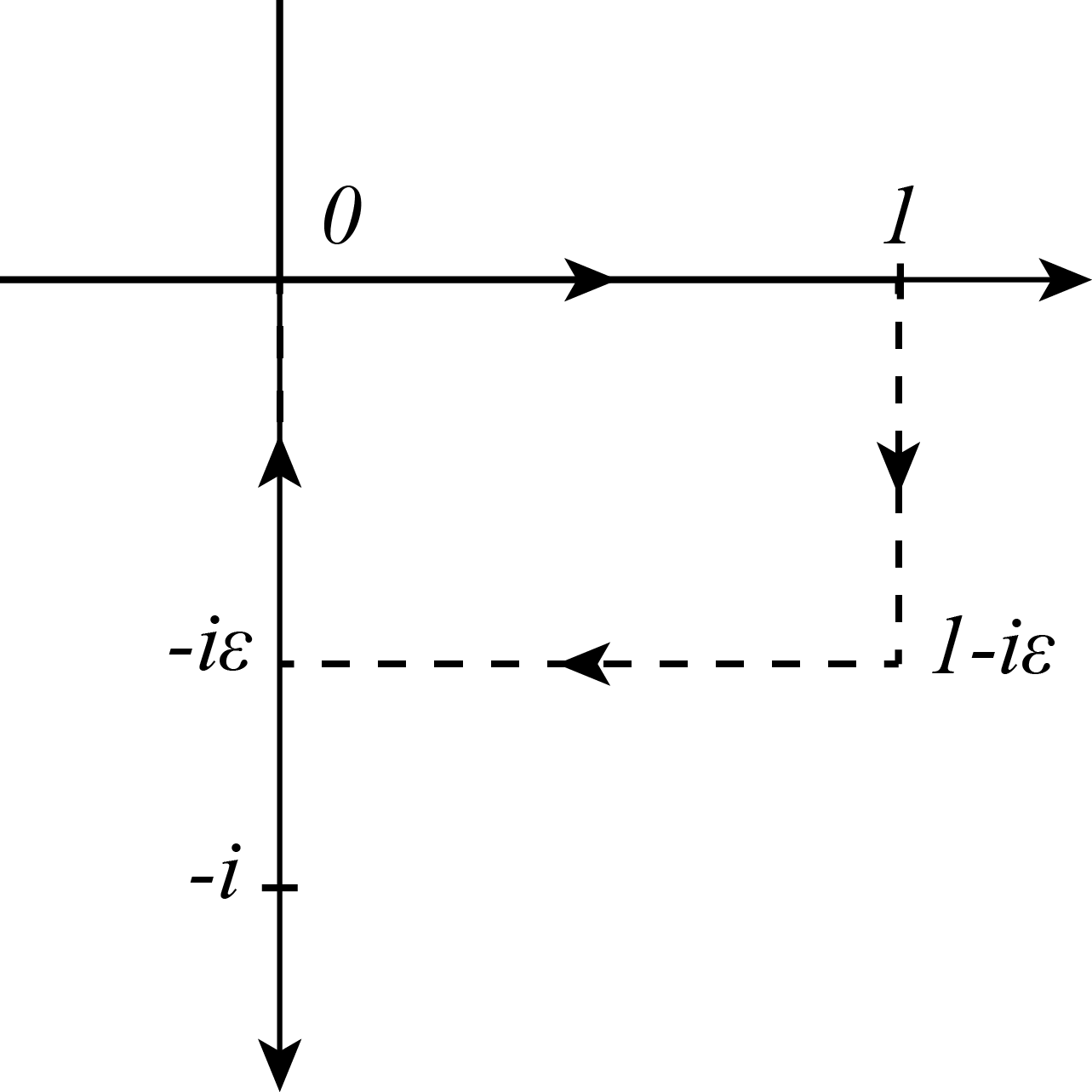}
\caption{Contour of integration}
\label{F6}
\end{figure}
We want to estimate equation \eqref{p4}. To that end assume that $p>0$ and $q\ge 0$ and because $f(x,y)$ is real analytic in $x$, we can integrate over the contour in Fig.\ \ref{F6} for some $\varepsilon>0$. By the Cauchy integral formula and the periodicity of $f(x,y)$, it follows that
\begin{equation}\label{p9}
\begin{aligned}
a(p,q)&=\int\limits_0^1\int\limits_0^1f(x,y)e^{-2\pi i (px+qy)}\,dx\,dy\\
&=e^{-2\pi\varepsilon p}\int\limits_0^1\int\limits_0^1f(x-i\varepsilon,y)e^{-2\pi i (px+qy)}\,dx\,dy\\
&=\mathcal{O}\left(e^{-2\pi \varepsilon p}\right).
\end{aligned}
\end{equation}
If $p<0$, then use the contour above reflected across the real axis to get
\begin{equation}\label{p10}
a(p,q)=\mathcal{O}\left(e^{-2\pi  \varepsilon |p|}\right).
\end{equation}
We can likewise assume $p\ge 0$ and $|q|>0$ and perform the same reasoning with respect to $y$ to obtain
\begin{equation}\label{p11}
a(p,q)=\mathcal{O}\left(e^{-2\pi \varepsilon(|p|+|q|)}\right).
\end{equation}
From this equation and equation \eqref{p8}, we conclude that 
\begin{equation}\label{p12}
\begin{aligned}
\frac{1}{mn}\sum_{j=0}^{m-1}\sum_{k=0}^{n-1}
f(x_j,y_k)&=a(0,0)+r(m,n),
\end{aligned}
\end{equation}
where 
\begin{equation}\label{p14}
\begin{aligned}
|r(m,n)|\le C\sideset{}{'}\sum_{(s,t)\in\Z^2}e^{-2\pi \varepsilon (m|s|+n|t|)}\le C_0e^{-2\pi \varepsilon(m+n)},\quad C_0>0.
\end{aligned}
\end{equation}
Returning to equation \eqref{p1} we have that
\begin{equation}\label{p13}
\frac{1}{mn}\ln\det A_i=\ln 2+\int\limits_0^1\int\limits_0^1 f(x,y)\,dx\,dy+\mathcal{O}\left(e^{-2\pi\varepsilon(m+n)}\right).
\end{equation}
From \eqref{p13} it follows that for each $i=1,2,3,4,$ we can write
\begin{equation}\label{15}
\det A_i=e^{mnF}\left(1+\mathcal{O}\left(e^{-c(m+n)}\right)\right).
\end{equation}
Finally, from equations \eqref{a1} and \eqref{PSTF} the claim follows.
\end{proof}

\section{Conclusion}

In this paper we establish the signs of the Pfaffians $\Pf A_i$ in the
dimer model on the triangular lattice on the torus. We prove that $\Pf A_1<0$,
while $\Pf A_i>0$ for $i=2,3,4$. Our proof is based on the Kasteleyn's
expressions for $\Pf A_i$ in terms of restricted partition sums $Z^{rs}$ and
on low weight expansions. As an application, we obtain an asymptotic expansion
of the partition function in the limit as $m,n\to\infty$.  It would be interesting to extend the Pfaffian Sign Theorem to other lattices on the torus.

As shown by Galluccio and Loebl \cite{GalluLoe}, Tesler \cite{Tes}, and Cimasoni and Reshetikhin \cite{CimResh}, the dimer models on an orientable  Riemann surface of genus $g$ is expressed as a linear algebraic combination of $2^{2g}$ Pfaffians. It is extended to non-orientable surfaces in the work of Tesler \cite{Tes}. The Pfaffian Sign Theorem in this general setting is another interesting open problem.

{\bf Acknowledgements.} The authors thank Barry McCoy and Dan Ramras for useful discussions, and the referee for a simplified proof of Theorem \ref{Conj}.

\begin{appendix}

\section{Proof of Lemma \ref{Co}} \label{appA}

We have that 
\begin{equation}\label{sqcup}
\sg\cup\sg'=\bigsqcup_{i=1}^{r}\ga_i.
\end{equation}
Since each vertex in $\ga_i$ is occupied by a dimer either from  $\sg$ or $\sg',$ and the dimers in $\sg\cup\sg'$ alternate, we conclude that each $\ga_i$ is of even length.
\par
By \eqref{Pfsign}, 
\begin{equation}
\Pf A=\sum_{\sg\in\Sg_{m,n}}\sgn(\sg)w(\sg). 
\end{equation} 
If we enumerate the vertices on $\Ga_{m,n}$, i.e. permute the set of vertices $V_{m,n},$ then by the well-known fact (see e.g. \cite{Godsil}) that for an arbitrary matrix $P$ of order $mn\times mn,$ 
\begin{equation}\label{fact}
\Pf(PAP^{T})=\det(P)\Pf(A),
\end{equation}
we get 
\begin{equation}
\Pf \rho(A)=(-1)^{\rho}\Pf(A),
\end{equation}
where $\rho$ is some permutation on $V_{m,n}.$ Here $\rho(A)$ denotes a matrix $A$ whose rows and columns have been permuted by $\rho.$ In other words,
\begin{equation}\label{identity}
\Pf A=\sum_{\sg\in\Sg_{m,n}}\sgn(\sg)w(\sg)=\sum_{\sg\in\Sg_{m,n}}(-1)^{\rho}[\sgn(\sg)]_{\rho}w(\sg),
\end{equation}
where $[\sgn(\sg)]_{\rho}$ indicates the sign of $\sg$ with respect to some new enumeration $\rho$ of vertices. From \eqref{identity}, we have that
\begin{equation}\label{invar}
\sgn(\sg)=(-1)^{\rho}[\sgn(\sg)]_{\rho}.
\end{equation}
If we take any two configurations $\sg$ and $\sg'$, then \eqref{invar} implies that \begin{equation}\label{rho}
\sgn(\sg)\cdot\sgn(\sg')=[\sgn(\sg)]_{\rho}\cdot[\sgn(\sg')]_{\rho},
\end{equation}
i.e. the sign of $\sg\cup\sg'$ is invariant under any renumeration of vertices.
\par
Let
\begin{equation}
  \rho = \left(\begin{matrix}
  1 & 2 & \cdots & n_1 & n_1+1 & \cdots & n_1+n_2 & \cdots & n_1+\ldots+n_{r-1}+n_r \\
  v_{1,1} & v_{1,2} & \cdots & v_{1,n_1} & v_{2,1} & \cdots & v_{2,n_2} & \cdots  &       v_{r,n_r}
\end{matrix}\right),
\end{equation}
where $v_{j,k}\in\{1,2,\ldots,mn\}$ denotes the $k$-th vertex of the $j$-th contour. Note that $\rho$ is a renumeration of vertices so that along each $\ga_i$ they are rearranged in a cyclical order, starting from one contour and continuing to the next one. Now, the underlying permutations $\pi(\sg)$ and $\pi(\sg')$ with respect to $\rho$ are then:
\begin{align}
& [\pi(\sg)]_{\rho}={\rm Id},\\
& [\pi(\sg')]_{\rho}=\prod_{i=1}^{r}C(\ga_i),
\end{align}
where \begin{equation}
  C(\ga_i) = \left(\begin{matrix}
  v_{i,1} & \cdots & v_{i,n_i} \\
  v_{i,2} & \cdots & v_{i,1} 
\end{matrix}\right).
\end{equation}
From this equation, equation \eqref{rho}, and the fact that each $\ga_i$ corresponds to a cycle of even length, Lemma \ref{Co} follows.

\section{Block Diagonalization of the Kasteleyn Matrices and the Double Product Formulae for $\det A_i,$ $i=1,2,3,4$} \label{appB}

In this Appendix we diagonalize the Kasteleyn matrices $A_i$ and derive the
double product formulae for $\det A_i$, $i=1,2,3,4$.
We write 
\begin{equation}\label{km1}
j=2j_0+r, \quad 0\le j_0\le m_0-1,\quad  m_0=\frac{m}{2}\,,\quad r=0,1,
\end{equation}
and we represent a vertex $x$ as
\begin{equation}\label{km2}
x=(j_0,k,r)\in \Z_{m_0}\oplus \Z_n\oplus\Z_2\,.
\end{equation}
Then the matrix elements of the Kasteleyn matrix $A_1$ are written in terms of the Kronecker $\de$-symbol $\de_{jk}$ as
\begin{equation}\label{km3}
\begin{aligned}
A_1(x,x')&=z_h\big[\de_{j_0j_0'}\de_{kk'}(\de_{r0}\de_{r'1}-\de_{r1}\de_{r'0})
+\de_{j_0+1,j_0'}\de_{kk'}\de_{r1}\de_{r'0}
-\de_{j_0,j_0'+1}\de_{kk'}\de_{r0}\de_{r'1}\big]\\
&+z_v\big[\de_{j_0j_0'}(\de_{k+1,k'}-\de_{k,k'+1})(\de_{r0}\de_{r'0}-\de_{r1}\de_{r'1})\big]\\
&+z_d\big[-\de_{j_0j_0'}\de_{k+1,k'}\de_{r0}\de_{r'1}+\de_{j_0j_0'}\de_{k,k'+1}\de_{r1}\de_{r'0} 
+\de_{j_0+1,j_0'}\de_{k+1,k'}\de_{r1}\de_{r'0}\\
&-\de_{j_0,j_0'+1}\de_{k,k'+1}\de_{r0}\de_{r'1} \big].
\end{aligned}
\end{equation}
It is convenient to write the matrix $A_1$ in terms of tensor products as follows. Denote by $I_n$ 
the identity $n\times n$ matrix and put
\begin{equation}\label{km4}
\begin{aligned}
J_{n}^+&=(\de_{k+1,k'})_{k,k'\in\Z_n}
=\begin{pmatrix}
0 & 1 & 0 & \ldots & 0 & 0\\
0 & 0 & 1 & \ldots & 0 & 0\\
0 & 0 & 0 & \ldots & 0 & 0\\
\vdots & \vdots & \vdots & \ddots & \vdots  & \vdots \\
0 & 0 & 0 & \ldots & 0 & 1 \\
1 & 0 & 0 & \ldots & 0 & 0
\end{pmatrix},\\ & {}\\
\quad J_{n}^-&=(\de_{k,k'+1})_{k,k'\in\Z_n}
=\begin{pmatrix}
0 & 0 & 0 & \ldots & 0 & 1\\
1 & 0 & 0 & \ldots & 0 & 0\\
0 & 1 & 0 & \ldots & 0 & 0\\
\vdots & \vdots & \vdots & \ddots & \vdots  & \vdots \\
0 & 0 & 0 & \ldots & 0 & 0 \\
0 & 0 & 0 & \ldots & 1 & 0
\end{pmatrix}.
\end{aligned}
\end{equation}
Then \eqref{km3} can be rewritten as follows:
\begin{equation}\label{km5}
\begin{aligned}
A_1&=z_h\left[I_{m_0}\otimes I_n\otimes
\begin{pmatrix}
0 & 1 \\
-1 & 0
\end{pmatrix}
+J_{m_0}^+\otimes I_n \otimes 
\begin{pmatrix}
0 & 0 \\
1 & 0
\end{pmatrix}
-J_{m_0}^-\otimes I_n \otimes 
\begin{pmatrix}
0 & 1 \\
0 & 0
\end{pmatrix}
\right]\\
&+z_v I_{m_0}\otimes (J_n^+-J_n^-)\otimes \begin{pmatrix}
1 & 0 \\
0 & -1
\end{pmatrix}
+z_d\bigg[-I_{m_0}\otimes J_n^+\otimes \begin{pmatrix}
0 & 1 \\
0 & 0
\end{pmatrix}\\
&+I_{m_0}\otimes J_n^-\otimes \begin{pmatrix}
0 & 0 \\
1 & 0
\end{pmatrix}
+J_{m_0}^+\otimes J_n^+\otimes \begin{pmatrix}
0 & 0 \\
1 & 0
\end{pmatrix}
-J_{m_0}^-\otimes J_n^-\otimes \begin{pmatrix}
0 & 1 \\
0 & 0
\end{pmatrix}
\bigg].
\end{aligned}
\end{equation}
We want to block-diagonalize the Kasteleyn matrix $A_1$ with $2\times 2$ blocks
along the diagonal. To that end introduce the vectors
\begin{equation}\label{km6}
\begin{aligned}
g_{\xi}=
\begin{pmatrix}
1  \\
e^{\frac{2\pi i \xi}{m_0}} \\
e^{\frac{4\pi i \xi}{m_0}} \\
\vdots \\
e^{\frac{2(m_0-1)\pi i \xi}{m_0}} 
\end{pmatrix}\in\R^{m_0},\qquad
h_{\eta}=
\begin{pmatrix}
1  \\
e^{\frac{2\pi i \eta}{n}} \\
e^{\frac{4\pi i \eta}{n}} \\
\vdots \\
e^{\frac{2(n-1)\pi i \eta}{n}} 
\end{pmatrix}\in\R^n,
\end{aligned}
\end{equation}
and
\begin{equation}\label{km7}
f_{\xi,\eta,r}=g_{\xi}\otimes h_{\eta}\otimes e_r\in \R^{mn},\quad \xi\in\Z_{m_0},\; \eta\in\Z_n,\; r\in \Z_2,
\end{equation}
with
\begin{equation}\label{km8}
e_0=\begin{pmatrix}1\\ 0\end{pmatrix}\in \R^2,\quad e_1=\begin{pmatrix}0\\ 1\end{pmatrix}\in\R^2.
\end{equation}
Observe that for $\xi\in\Z_{m_0}$ and $\eta\in\Z_n$, $g_{\xi}$ is an eigenfunction of the matrices 
$J_{m_0}^{\pm}$ and $h_{\eta}$ of the ones $J_n^{\pm}$. Namely, 
\begin{equation}\label{km9}
J_{m_0}^{\pm}g_{\xi}=e^{\pm\frac{2\pi i\xi}{m_0}}g_{\xi},\quad
J_{n}^{\pm}h_{\eta}=e^{\pm\frac{2\pi i\eta}{n}}h_{\eta},\quad \xi\in\Z_{m_0},\; \eta\in\Z_n.
\end{equation}
By \eqref{km5}, this implies that
\begin{equation}\label{km10}
A_1 f_{\xi,\eta,0}=\la_{00}f_{\xi,\eta,0}+\la_{10}f_{\xi,\eta,1},\quad
A_1 f_{\xi,\eta,1}=\la_{01}f_{\xi,\eta,0}+\la_{11}f_{\xi,\eta,1},
\end{equation}
with
\begin{equation}\label{km11}
\begin{aligned}
\la_{00}&=2iz_v\sin\frac{2\pi \eta}{n},\quad 
\la_{01}=z_h(1-e^{-\frac{2\pi i\xi}{m_0}})+z_d\big(-e^{\frac{2\pi i\eta}{n}}
-e^{-\frac{2\pi i\xi}{m_0}-\frac{2\pi i\eta}{n}}\big),\\
\la_{10}&=z_h(-1+e^{\frac{2\pi i\xi}{m_0}})+z_d\big(e^{-\frac{2\pi i\eta}{n}}
+e^{\frac{2\pi i\xi}{m_0}+\frac{2\pi i\eta}{n}}\big),\quad
\la_{11}=-2iz_v \sin\frac{2\pi \eta}{n}.
\end{aligned}
\end{equation}
From \eqref{km10},
\begin{equation}\label{km12}
A_1 \begin{pmatrix}
f_{\xi,\eta,0} \\
f_{\xi,\eta,1}
\end{pmatrix}=\La_{\xi, \eta}
\begin{pmatrix}
f_{\xi,\eta,0} \\
f_{\xi,\eta,1}
\end{pmatrix},\quad 
\La_{\xi, \eta}=\begin{pmatrix}
\la_{00} & \la_{01} \\
\la_{10} & \la_{11}
\end{pmatrix},\quad \xi\in\Z_{m_0},\; \eta\in\Z_n.
\end{equation}
Observe that $\La_{\xi,\eta}$ is an anti-Hermitian matrix so that
\begin{equation}\label{k11a}
\La_{\xi,\eta}^*=-\La_{\xi,\eta},
\end{equation}
and also
\begin{equation}\label{k11b}
\overline{\La_{\xi,\eta}}=\La_{-\xi,-\eta}\,,\quad \Tr \La_{\xi,\eta}=0.
\end{equation}
The matrix 
\begin{equation}\label{k12}
\La=\left(\La_{\xi,\eta}\de_{\xi,\xi'}\de_{\eta,\eta'}\right)_{\xi,\xi'\in \Z_{m_0};\, \eta,\eta'\in\Z_n}
\end{equation}
is a block-diagonal matrix, which is a  block-diagonalization of $A_1$, so that
\begin{equation}\label{k13}
A_1=U\La U^{-1}, \quad U=U_1\otimes U_2\otimes I_2,\quad U_1=\left( e^{\frac{2\pi i j\xi}{m_{0}}}\right)_{j,\xi=0}^{m_0-1},\quad
U_2=\left( e^{\frac{2\pi ik\eta}{n} }\right)_{k,\eta=0}^{n-1}.
\end{equation}
This implies that
\begin{equation}\label{km13}
\det A_1 =\prod_{\xi=0}^{\frac{m}{2}-1}\prod_{\eta=0}^{n-1}\det \La_{\xi,\eta}.
\end{equation}
It remains to calculate $\det \La_{\xi,\eta}$. From \eqref{km11} we have that
\begin{equation}\label{km14}
\begin{aligned}
\det \La_{\xi,\eta}&=\la_{00}\la_{11}-\la_{01}\la_{10}\\
&=4z_v^2 \sin\frac{2\pi \eta}{n}-
\big[z_h(1-e^{-\frac{2\pi i\xi}{m_0}})+z_d\big(-e^{\frac{2\pi i\eta}{n}}
-e^{-\frac{2\pi i\xi}{m_0}-\frac{2\pi i\eta}{n}}\big)\big]\\
&\times
\big[z_h(-1+e^{\frac{2\pi i\xi}{m_0}})+z_d\big(e^{-\frac{2\pi i\eta}{n}}
+e^{\frac{2\pi i\xi}{m_0}+\frac{2\pi i\eta}{n}}\big)\big]\\
&=4z_h^2 \sin^2\frac{2\pi \xi}{m}+4z_v^2 \sin^2\frac{2\pi \eta}{n}
+4z_d^2\cos^2\left(\frac{2\pi \xi}{m}+\frac{2\pi \eta}{n}\right),
\end{aligned}
\end{equation}
hence
\begin{equation}\label{km15}
\det A_1 =\prod_{\xi=0}^{\frac{m}{2}-1}\prod_{\eta=0}^{n-1}\left[4z_h^2 \sin^2\frac{2\pi \xi}{m}+4z_v^2 \sin^2\frac{2\pi \eta}{n}
+4z_d^2\cos^2\left(\frac{2\pi \xi}{m}+\frac{2\pi \eta}{n}\right)\right].
\end{equation}
This proves the double product formula for $\det A_1$.

Consider now $\det A_2$. The matrix elements of the matrix $ A_2$ are
\begin{equation}\label{k2m1}
\begin{aligned}
A_2(x,x')&=z_h\big[\de_{j_0j_0'}\de_{kk'}(\de_{r0}\de_{r'1}-\de_{r1}\de_{r'0})
+\de_{j_0+1,j_0'}\de_{kk'}\de_{r1}\de_{r'0}
-\de_{j_0,j_0'+1}\de_{kk'}\de_{r0}\de_{r'1}\big]\\
&+z_v\big[\de_{j_0j_0'}\big((-1)^{\de_{k,n-1}}\de_{k+1,k'}-(-1)^{\de_{k',n-1}}\de_{k,k'+1}\big)(\de_{r0}\de_{r'0}-\de_{r1}\de_{r'1})\big]\\
&+z_d\big[-\de_{j_0j_0'}(-1)^{\de_{k,n-1}}\de_{k+1,k'}\de_{r0}\de_{r'1}+\de_{j_0j_0'}(-1)^{\de_{k',n-1}}\de_{k,k'+1}\de_{r1}\de_{r'0}\\ 
&+\de_{j_0+1,j_0'}(-1)^{\de_{k,n-1}}\de_{k+1,k'}\de_{r1}\de_{r'0}
-\de_{j_0,j_0'+1}(-1)^{\de_{k',n-1}}\de_{k,k'+1}\de_{r0}\de_{r'1} \big],
\end{aligned}
\end{equation}
and the matrix $A_2$ can be written as
\begin{equation}\label{k2m2}
\begin{aligned}
A_2&=z_h\bigg[I_{m_0}\otimes I_n\otimes
\begin{pmatrix}
0 & 1 \\
-1 & 0
\end{pmatrix}
+J_{m_0}^+\otimes I_n \otimes 
\begin{pmatrix}
0 & 0 \\
1 & 0
\end{pmatrix}
-J_{m_0}^-\otimes I_n \otimes 
\begin{pmatrix}
0 & 1 \\
0 & 0
\end{pmatrix}
\bigg]\\
&+z_v I_{m_0}\otimes (K_n^+-K_n^-)\otimes \begin{pmatrix}
1 & 0 \\
0 & -1
\end{pmatrix}
+z_d\bigg[-I_{m_0}\otimes K_n^+\otimes \begin{pmatrix}
0 & 1 \\
0 & 0
\end{pmatrix}\\
&+I_{m_0}\otimes K_n^-\otimes \begin{pmatrix}
0 & 0 \\
1 & 0
\end{pmatrix}
+J_{m_0}^+\otimes K_n^+\otimes \begin{pmatrix}
0 & 0 \\
1 & 0
\end{pmatrix}
-J_{m_0}^-\otimes K_n^-\otimes \begin{pmatrix}
0 & 1 \\
0 & 0
\end{pmatrix}
\bigg],
\end{aligned}
\end{equation}
with
\begin{equation}\label{k2m3}
\begin{aligned}
K_{n}^+&
=\begin{pmatrix}
0 & 1 & 0 & \ldots & 0 & 0\\
0 & 0 & 1 & \ldots & 0 & 0\\
0 & 0 & 0 & \ldots & 0 & 0\\
\vdots & \vdots & \vdots & \ddots & \vdots  & \vdots \\
0 & 0 & 0 & \ldots & 0 & 1 \\
-1 & 0 & 0 & \ldots & 0 & 0
\end{pmatrix},
\quad K_{n}^-&
=\begin{pmatrix}
0 & 0 & 0 & \ldots & 0 & -1\\
1 & 0 & 0 & \ldots & 0 & 0\\
0 & 1 & 0 & \ldots & 0 & 0\\
\vdots & \vdots & \vdots & \ddots & \vdots  & \vdots \\
0 & 0 & 0 & \ldots & 0 & 0 \\
0 & 0 & 0 & \ldots & 1 & 0
\end{pmatrix}.
\end{aligned}
\end{equation}
Observe that the vectors 
\begin{equation}\label{k2m4}
\begin{aligned}
h_{\eta+\frac{1}{2}}=
\begin{pmatrix}
1  \\
e^{\frac{2\pi i \left(\eta+\frac{1}{2}\right)}{n}} \\
e^{\frac{4\pi i  \left(\eta+\frac{1}{2}\right)}{n}} \\
\vdots \\
e^{\frac{2(n-1)\pi i  \left(\eta+\frac{1}{2}\right)}{n}}
\end{pmatrix}\in\R^{n},\qquad \eta\in\Z_n,
\end{aligned}
\end{equation}
are eigenvectors of the matrices $K_n^{\pm}$.  Namely, 
\begin{equation}\label{k2m5}
K_{n}^{\pm}h_{\eta+\frac{1}{2}}=e^{\pm\frac{2\pi i\left(\eta+\frac{1}{2}\right)}{n}}h_{\eta+\frac{1}{2}},
\quad \xi\in\Z_{m_0},\; \eta\in\Z_n.
\end{equation}
Consider the vectors
\begin{equation}\label{k2m6}
f_{\xi,\eta+\frac{1}{2},r}=g_{\xi}\otimes h_{\eta+\frac{1}{2}}\otimes e_r,\quad \xi\in\Z_{m_0},\; \eta\in\Z_n,\; r\in \Z^2.
\end{equation}
Then
\begin{equation}\label{k2m7}
A_2 f_{\xi,\eta+\frac{1}{2},0}=\la_{00}f_{\xi,\eta+\frac{1}{2},0}+\la_{10}f_{\xi,\eta+\frac{1}{2},1},\quad
A_2 f_{\xi,\eta+\frac{1}{2},1}=\la_{01}f_{\xi,\eta+\frac{1}{2},0}+\la_{11}f_{\xi,\eta+\frac{1}{2},1},
\end{equation}
with
\begin{equation}\label{k2m8}
\begin{aligned}
\la_{00}&=2iz_v\sin\frac{2\pi \left(\eta+\frac{1}{2}\right)}{n},\quad 
\la_{11}=-2iz_v \sin\frac{2\pi \left(\eta+\frac{1}{2}\right)}{n}\,,\\
\la_{01}&=z_h(1-e^{-\frac{2\pi i\xi}{m_0}})+z_d\big(-e^{\frac{2\pi i\left(\eta+\frac{1}{2}\right)}{n}}
-e^{-\frac{2\pi i\xi}{m_0}-\frac{2\pi i\left(\eta+\frac{1}{2}\right)}{n}}\big),\\
\la_{10}&=z_h(-1+e^{\frac{2\pi i\xi}{m_0}})+z_d\big(e^{-\frac{2\pi i\left(\eta+\frac{1}{2}\right)}{n}}
+e^{\frac{2\pi i\xi}{m_0}+\frac{2\pi i\left(\eta+\frac{1}{2}\right)}{n}}\big).
\end{aligned}
\end{equation}
The matrix 
\begin{equation}\label{k12a}
\La=\left(\La_{\xi,\eta}\de_{\xi,\xi'}\de_{\eta,\eta'}\right)_{\xi,\xi'\in \Z_{m_0};\, \eta,\eta'\in\Z_n}
\end{equation}
is now a  block-diagonalization of $A_2$ with
\begin{equation}\label{k13a}
A_2=U\La U^{-1}, \quad U=U_1\otimes U_2\otimes I_2,\quad U_1=\left( e^{\frac{2\pi i j\xi}{m_{0}}}\right)_{j,\xi=0}^{m_0-1},\quad
U_2=\left( e^{\frac{2\pi ik\left(\eta+\frac{1}{2}\right)}{n} }\right)_{k,\eta=0}^{n-1}.
\end{equation}
From here it follows that
\begin{equation}\label{k2m9}
\begin{aligned}
\det A_2 &=\prod_{\xi=0}^{\frac{m}{2}-1}\prod_{\eta=0}^{n-1}\Bigg[4z_h^2 \sin^2\frac{2\pi \xi}{m}
+4z_v^2 \sin^2\frac{2\pi \left(\eta+\frac{1}{2}\right)}{n} \\
&+4z_d^2\cos^2\left(\frac{2\pi \xi}{m}+\frac{2\pi \left(\eta+\frac{1}{2}\right)}{n}\right)\Bigg].
\end{aligned}
\end{equation}
Let us turn to $\det A_3$.

The matrix elements of the matrix $A_3$ are 
\begin{equation}\label{k3m1}
\begin{aligned}
A_3(x,x')&=z_h\big[\de_{j_0j_0'}\de_{kk'}(\de_{r0}\de_{r'1}-\de_{r1}\de_{r'0})
+(-1)^{\de_{j_0,m_0-1}}\de_{j_0+1,j_0'}\de_{kk'}\de_{r1}\de_{r'0}\\
&-(-1)^{\de_{j_0',m_0-1}}\de_{j_0,j_0'+1}\de_{kk'}\de_{r0}\de_{r'1}\big]\\
&+z_v\big[\de_{j_0j_0'}(\de_{k+1,k'}-\de_{k,k'+1})(\de_{r0}\de_{r'0}-\de_{r1}\de_{r'1})\big]\\
&+z_d\big[-\de_{j_0j_0'}\de_{k+1,k'}\de_{r0}\de_{r'1}+\de_{j_0j_0'}\de_{k,k'+1}\de_{r1}\de_{r'0} \\
&+(-1)^{\de_{j_0,m_0-1}}\de_{j_0+1,j_0'}\de_{k+1,k'}\de_{r1}\de_{r'0}
-(-1)^{\de_{j_0',m_0-1}}\de_{j_0,j_0'+1}\de_{k,k'+1}\de_{r0}\de_{r'1} \big],
\end{aligned}
\end{equation}
hence
\begin{equation}\label{k3m2}
\begin{aligned}
A_3&=z_h\Big[I_{m_0}\otimes I_n\otimes
\begin{pmatrix}
0 & 1 \\
-1 & 0
\end{pmatrix}
+K_{m_0}^+\otimes I_n \otimes 
\begin{pmatrix}
0 & 0 \\
1 & 0
\end{pmatrix}
-K_{m_0}^-\otimes I_n \otimes 
\begin{pmatrix}
0 & 1 \\
0 & 0
\end{pmatrix}
\Big]\\
&+z_v I_{m_0}\otimes (J_n^+-J_n^-)\otimes \begin{pmatrix}
1 & 0 \\
0 & -1
\end{pmatrix}
+z_d\bigg[-I_{m_0}\otimes J_n^+\otimes \begin{pmatrix}
0 & 1 \\
0 & 0
\end{pmatrix}\\
&+I_{m_0}\otimes J_n^-\otimes \begin{pmatrix}
0 & 0 \\
1 & 0
\end{pmatrix}
+K_{m_0}^+\otimes J_n^+\otimes \begin{pmatrix}
0 & 0 \\
1 & 0
\end{pmatrix}
-K_{m_0}^-\otimes J_n^-\otimes \begin{pmatrix}
0 & 1 \\
0 & 0
\end{pmatrix}
\bigg].
\end{aligned}
\end{equation}
Now, the matrix 
\begin{equation}\label{k12b}
\La=\left(\La_{\xi,\eta}\de_{\xi,\xi'}\de_{\eta,\eta'}\right)_{\xi,\xi'\in \Z_{m_0};\, \eta,\eta'\in\Z_n}
\end{equation}
is a  block-diagonalization of $A_3$ with
\begin{equation}\label{k13a1}
A_3=U\La U^{-1}, \quad U=U_1\otimes U_2\otimes I_2,\quad U_1=\left( e^{\frac{2\pi i j\left(\xi+\frac{1}{2}\right)}{m_{0}}}\right)_{j,\xi=0}^{m_0-1},\quad
U_2=\left( e^{\frac{2\pi ik\eta}{n} }\right)_{k,\eta=0}^{n-1}.
\end{equation}
Hence,
\begin{equation}\label{k3m3}
\begin{aligned}
\det A_3&=\prod_{\xi=0}^{\frac{m}{2}-1}\prod_{\eta=0}^{n-1}\Bigg[4z_h^2 \sin^2\frac{2\pi \left(\xi+\frac{1}{2}\right)}{m}
+4z_v^2 \sin^2\frac{2\pi \eta}{n}\\
&+4z_d^2\cos^2\left(\frac{2\pi \left(\xi+\frac{1}{2}\right)}{m}+\frac{2\pi \eta}{n}\right)\Bigg].
\end{aligned}
\end{equation}
The double product formula for $\det A_4$ is proved similarly:
\begin{equation}\label{k3m4}
\begin{aligned}
\det A_4&=\prod_{\xi=0}^{\frac{m}{2}-1}\prod_{\eta=0}^{n-1}\Bigg[4z_h^2 \sin^2\frac{2\pi \left(\xi+\frac{1}{2}\right)}{m}
+4z_v^2 \sin^2\frac{2\pi \left(\eta+\frac{1}{2}\right)}{n}\\
&+4z_d^2\cos^2\left(\frac{2\pi \left(\xi+\frac{1}{2}\right)}{m}+\frac{2\pi \left(\eta+\frac{1}{2}\right)}{n}\right)\Bigg].
\end{aligned}
\end{equation}

\section{Numerical Data for the Pfaffians $A_i$}\label{numerics}

In this Appendix we present numerical data for the Pfaffians $A_i$ on the $m\times n$ lattices on the torus for different values $m$ and $n$. It is interesting to compare these data with the asymptotics of the Pfaffians $A_i$, obtained in Sections \ref{A2} and \ref{A1} above, and also with the identities $\Pf A_1=-\Pf A_2$, $\Pf A_3=\Pf A_4$ for odd $n$, proven in Section \ref{Id}. 

 The Pfaffians $A_i$ for $m=4$, $n=3$.
\[
\begin{aligned}
\Pf A_1&=-4z_hz_d(3z_v^2+z_d^2)(4z_h^2+3z_v^2+3z_d^2),\\
\Pf A_2&=4z_hz_d(3z_v^2+z_d^2)(4z_h^2+3z_v^2+3z_d^2),\\
\Pf A_3&=2(z_h^2+z_d^2)[(2z_h^2+3z_v^2)(2z_h^2+3z_v^2+4z_d^2)+z_d^4],\\
\Pf A_4&=2(z_h^2+z_d^2)[(2z_h^2+3z_v^2)(2z_h^2+3z_v^2+4z_d^2)+z_d^4].
\end{aligned}
\]

 The Pfaffians $A_i$ for $m=4$, $n=4$.
\[
\begin{aligned}
\Pf A_1&=-256z_h^2z_v^2z_d^2(z_h^2+z_v^2+z_d^2),\\
\Pf A_2&=16(z_v^2+z_d^2)^2(2z_h^2+z_v^2+z_d^2)^2,\\
\Pf A_3&=16(z_v^2+z_d^2)^2(z_h^2+2z_v^2+z_d^2)^2,\\
\Pf A_3&=16(z_v^2+z_d^2)^2(z_h^2+z_v^2+2z_d^2)^2.
\end{aligned}
\]

 The Pfaffians $A_i$ for $m=4$, $n=6$.
\[
\begin{aligned}
\Pf A_1&=-16z_h^2z_d^2(4z_h^2+3z_v^2+3z_d^2)^2(3z_v^2+z_d^2)^2,\\
\Pf A_2&=16z_v^2(4z_h^2+z_v^2+z_d^2)^2(z_h^2+z_v^2+z_d^2)(z_v^2+3z_d^2)^2,\\
\Pf A_3&=4(z_d^2 + z_h^2)^2(z_d^4 + 4z_d^2 (2z_h^2 + 3z_v^2) + (2z_h^2 + 3z_v^2)^2)^2,\\
\Pf A_4&=4(z_d^2 + z_h^2 + 2z_v^2)^2(z_d^4 + 8z_d^2z_h^2 + 4z_h^4 + 4z_d^2z_v^2 + 
   4z_h^2z_v^2 + z_v^4)^2.
\end{aligned}
\]

 The Pfaffians $A_i$ for $m=4$, $n=8$.
\[
\begin{aligned}
\Pf A_1&=-4096z_d^2z_h^2z_v^2(z_d^2 + z_v^2)^2(z_d^2 + z_h^2 + z_v^2)(z_d^2 + 2z_h^2 + z_v^2)^2,\\
\Pf A_2&=16(z_d^4 + 6z_d^2z_v^2 + z_v^4)^2(z_d^4 + 8z_h^4 + 8z_h^2 z_v^2 + z_v^4 + 
   2z_d^2 (4z_h^2 + z_v^2))^2,\\
\Pf A_3&=256(z_d^2 + z_h^2)^2(z_h^2 + z_v^2)^2(2z_d^2 + z_h^2 + z_v^2)^2(z_d^2 + z_h^2 + 
   2z_v^2)^2,\\
\Pf A_4&=16(z_d^4 + 2z_h^4 + 4z_h^2z_v^2 + z_v^4 + 2z_d^2(2z_h^2 + z_v^2))^2(z_d^4 + 
   2z_h^4 + 4z_h^2z_v^2 + z_v^4 + z_d^2 (4z_h^2 + 6z_v^2))^2.
\end{aligned}
\]

 The Pfaffians $A_i$ for $m=6$, $n=6$.
\[
\begin{aligned}
\Pf A_1&=-4z_d^2(z_d^2 + 3z_h^2)^2(z_d^2 + 3z_v^2)^2 (z_d^2 + 
   3 (z_h^2 + z_v^2))^2 (4z_d^2 + 3 (z_h^2 + z_v^2))^2,\\
\Pf A_2&=4 z_v^2 (3 z_d^2 + z_v^2)^2 (3 z_h^2 + z_v^2)^2 (3 (z_d^2 + z_h^2) + 
   z_v^2)^2 (3 (z_d^2 + z_h^2) + 4 z_v^2)^2,\\
\Pf A_3&=4 z_h^2 (3 z_d^2 + z_h^2)^2 (z_h^2 + 3 z_v^2)^2 (3 z_d^2 + z_h^2 + 
   3 z_v^2)^2 (3 z_d^2 + 4 z_h^2 + 3 z_v^2)^2,\\
\Pf A_4&=4 (z_d^2 + z_h^2 + z_v^2)^3 (4 z_d^2 + z_h^2 + z_v^2)^2 (z_d^2 + 4 z_h^2 + 
   z_v^2)^2 (z_d^2 + z_h^2 + 4 z_v^2)^2.
\end{aligned}
\]

The Pfaffians $A_i$ for $m=8$, $n=8$.
 
\[
\begin{aligned}
\Pf A_1&=-1048576  z_h^2z_v^2z_d^2 (z_d^2 + z_h^2)^2  (z_d^2 + z_v^2)^2 (z_h^2 + z_v^2)^2 (z_d^2 + z_h^2 + z_v^2) (2 z_d^2 + z_h^2 + z_v^2)^2 \\ &\qquad \times(z_d^2 + 2 z_h^2 + z_v^2)^2 (z_d^2 +z_h^2 + 2z_v^2)^2,\\
\Pf A_2&=256(z_d^4 +6z_d^2z_v^2 + z_v^4)^2(z_d^4 + 4z_d^2z_h^2 + 2z_h^4 + 2z_d^2z_v^2 + 4z_h^2z_v^2 +z_v^4)^2\\ &\qquad \times(z_d^4 + 4z_d^2z_h^2 + 2z_h^4 + 6z_d^2z_v^2 + 4z_h^2z_v^2 + z_v^4)^2(z_d^4 + 8z_d^2z_h^2 + 8z_h^4 + 2z_d^2z_v^2 + 8z_h^2z_v^2 +z_v^4)^2,\\
\Pf A_3&=256(z_d^4 + 6z_d^2z_h^2 + z_h^4)^2(z_d^4 + 2z_d^2z_h^2 + z_h^4 + 4z_d^2z_v^2 + 4z_h^2 z_v^2 + 2z_v^4)^2\\ &\qquad \times(z_d^4 + 6z_d^2z_h^2 + z_h^4 + 4z_d^2z_v^2 + 4z_h^2z_v^2 + 2z_v^4)^2 (z_d^4 + 2z_d^2z_h^2 + z_h^4 + 8z_d^2z_v^2 + 8z_h^2z_v^2 + 8z_v^4)^2,\\
\Pf A_4&=256 (2 z_d^4 + 4 z_d^2 z_h^2 + z_h^4 + 4 z_d^2 z_v^2 + 2 z_h^2 z_v^2 + z_v^4)^2 (8 z_d^4 + 8 z_d^2 z_h^2 + z_h^4 + 8 z_d^2 z_v^2 + 2 z_h^2 z_v^2 +z_v^4)^2 \\
&\qquad \times(z_h^4 + 6 z_h^2 z_v^2 + z_v^4)^2 (2 z_d^4 + 4 z_d^2 z_h^2 + z_h^4 +4 z_d^2 z_v^2 + 6 z_h^2 z_v^2 + z_v^4)^2.
\end{aligned}
\]

\end{appendix}

\end{document}